\documentclass[11pt]{article}

\usepackage{amsfonts,amssymb,amsmath,latexsym,xfrac,ae,aecompl}

\newcommand{\F}{\vspace*{\smallskipamount}}
\newcommand{\FF}{\vspace*{\medskipamount}}
\newcommand{\FFF}{\vspace*{\bigskipamount}}
\newcommand{\B}{\vspace*{-\smallskipamount}}

\newcommand{\BBB}{\vspace*{-\bigskipamount}}

\newcommand{\cA}{\mathcal{A}}

\newcommand{\cE}{\mathcal{E}}

\newcommand{\cI}{\mathcal{I}}
\newcommand{\cJ}{\mathcal{J}}

\newcommand{\cO}{\mathcal{O}}

\newcommand{\mE}{\mathbb{E}\,}

\newcommand{\Item}{\B\item}
\newcommand{\Paragraph}[1]{\BBB\paragraph{#1}}
\newcommand{\remove}[1]{}

\newlength{\pagewidth}
\newlength{\captionwidth}
\newcommand{\qed}{\hfill $\square$ \smallbreak}
\newenvironment{proof}{\noindent{\bf Proof:}}{\qed}

\newtheorem{theorem}{Theorem}
\newtheorem{lemma}{Lemma}

\newtheorem{corollary}{Corollary}
\newtheorem{proposition}{Proposition}

\begin{document}

\baselineskip           	3ex
\parskip                	1ex

\title{Anonymous  Processors with Synchronous Shared Memory \footnotemark[1]\FFF\FFF\FFF}

\author{	Bogdan S. Chlebus\,\footnotemark[2]   	
		\and
		Gianluca De Marco\,\footnotemark[3]
		\and
		Muhammed Talo\,\footnotemark[2]}

\footnotetext[1]{The work of the first author was supported by the National Science Foundation under Grant 1016847.}

\footnotetext[2]{Department of Computer Science and Engineering, 
			University of Colorado Denver, 
			Denver, Colorado 80217, USA.}

\footnotetext[3]{Dipartimento di Informatica,
                  Universit\`a degli Studi di Salerno,
                  Fisciano, 84084 Salerno, Italy.}

\date{}

\maketitle

\vfill


\begin{abstract}
We consider synchronous distributed systems in which anonymous processors communicate by shared read-write variables.
The goal is to have all the processors assign unique names to themselves.
We consider the instances of this problem determined by whether the number $n$ is known or not, and whether concurrently attempting to write distinct values into the same  memory cell is allowed or not, and whether the number of shared variables is a constant independent of~$n$ or it is unbounded.
For known~$n$, we give Las Vegas algorithms that operate in the optimum expected time, as determined by the amount of available shared memory,  and use the optimum $\cO(n\log n)$ expected number of random bits.
For unknown~$n$, we give Monte Carlo algorithms that produce correct output upon termination with probabilities that are $1-n^{-\Omega(1)}$, which is best possible when terminating almost surely and using $\cO(n\log n)$  random bits.

\vfill

~

\end{abstract}

\vfill

\thispagestyle{empty}

\setcounter{page}{0}											

\newpage

\section{Introduction} 

\label{sec:introduction}

We consider a distributed system in which some $n$ processors communicate using read-write shared memory.
It is assumed that operations performed on shared memory occur synchronously, in that executions of algorithms are structured as sequences of globally synchronized rounds.
The model of synchronous systems with read-write registers is known as the Parallel Random Access Machine (PRAM).
It is a generalization of the Random Access Machine model of sequential computation~\cite{Boas90} to the realm of synchronous concurrent processing.

We study the problem of assigning unique integer names from the interval $[1,n]$ to the $n$  processors of a PRAM, when originally the processors do not have distinct identifiers.
This task is called \emph{naming} and is understood such that all the processors cooperate  by executing a distributed algorithm to assign unique names to themselves.
We assume that the original anonymous processors do not have any feature facilitating identification or distinguishing one from another.
When processors of a distributed/parallel system are anonymous then the task of assigning  unique identifiers to all processors is a key step in making the system fully operational, because names are needed for executing deterministic algorithms.

The task to assign unique names to anonymous processes by themselves in distributed systems can be considered as a stage in either building such systems or making them fully operational.
Correspondingly, this may be categorized as either an architectural challenge or an algorithmic one.
For example, tightly synchronized message passing systems are typically considered under the assumption that processors are already equipped with unique identifiers.
This is because such systems impose strong demands on the architecture and the task of assigning  identifiers to processors is modest when compared to providing synchrony.
Similarly, when synchronous parallel machines are designed, then processors may be identified by how they are attached to the underlying communication network.
In contrast to that, PRAM is a virtual model in which processors communicate via shared memory; see an exposition of PRAM as a programming environment given by Keller et al.~\cite{KellerKL-book-2001}.
This model does not assume any relation between the shared memory and the processors that would be conducive to identifying processors.

Distributed systems with shared read-write registers are usually considered to be asynchronous.
Synchrony in such environments can be added by simulation rather than by a supportive architecture or an underlying communication network.
Processes do not need to be hardware nodes, instead, they can be virtual computing agents.
When a synchronous PRAM is considered, as  obtained by a simulation, then the underlying system architecture does not facilitate identifying processors, and so we do not necessarily expect that processors are equipped with distinct identifiers at the start of a simulation.

We view PRAM as an abstract construct which provides a  distributed environment to develop algorithms with multiple agents/processors working concurrently; see Vishkin~\cite{Vishkin11} for a comprehensive exposition of PRAM as a vehicle facilitating parallel programing and harnessing the power of multi-core computer architectures.
Assigning  names to processors by themselves in a distributed manner is a plausible stage in an algorithmic development of such environments, as it cannot be delegated to the stage of building hardware of a parallel machine.

We consider two categories of naming problems depending on how much shared memory is available for a PRAM.
In one case, the memory is bounded, in that just a constant number of memory cells is available.
This means that the amount of memory is independent from the number of processors~$n$ but as large as needed in  an algorithm's design.
In the other case, the number of shared memory cells is unbounded, in the sense that it unlimited in principle but how much is actually used by an algorithm depends on~$n$.
When it is assumed that an unbounded amount of memory cells is available, then the expected number of  memory cells that are actually used may be considered as a performance metric.

Independently of the amount of shared memory available, we consider two versions of the naming problem, determined by the semantics of concurrent writing in the underlying model of computation.
This is represented  by the corresponding  PRAM variants, which are either the Arbitrary PRAM or the Common PRAM.
The Arbitrary PRAM allows to attempt to write concurrently distinct values  into a register, and an arbitrary one of them gets written.
The Common PRAM variant allows only equal values to be concurrently written into a register.

Randomized algorithms are typically categorized as either Las Vegas or Monte Carlo; this categorization is understood as follows.
A randomized algorithm is Las Vegas when it terminates almost surely and the algorithm returns a correct output upon termination.
A randomized algorithm is Monte Carlo when it terminates almost surely and an incorrect output may be produced upon termination, but the probability of error converges to zero with the number of processors growing unbounded.

\newcommand{\RB}{\raisebox{3ex}{~}}
\newcommand{\LB}{\raisebox{-1.5ex}{~}}

\begin{table}[t]
\begin{center}
\begin{tabular}{|l| c |c |l |}
\hline
\RB \LB
PRAM Model&Memory&Time&Algorithm\\
\hline
\hline
\RB \LB
Arbitrary&$\cO(1)$&$\cO(n)$&\textsc{Arbitrary-Bounded-LV} in Section~\ref{sec:arbitrary-bounded-LV}\\
\hline
\RB \LB
Arbitrary&$\cO(n/\log n)$&$\cO(\log n)$&\textsc{Arbitrary-Unbounded-LV} in Section~\ref{sec:arbitrary-unbounded-LV}\\
\hline
\RB \LB
Common&$\cO(1)$&$\cO(n\log n)$&\textsc{Common-Bounded-LV} in Section~\ref{sec:common-bounded-LV}\\
\hline
\RB \LB
Common&$\cO(n)$&$\cO(\log n)$&\textsc{Common-Unbounded-LV} in Section~\ref{sec:common-unbounded-LV}\\
\hline
\end{tabular}
\parbox{\pagewidth}{\FF\caption{\label{tab:summary-LV} 
Four naming problems for known $n$, as determined by the PRAM model and the available amount of shared memory, with the respective performance bounds. 
All algorithms are Las Vegas.
}}
\end{center}
\end{table}

We say that a parameter of an algorithmic problem is \emph{known} when it can be used in a code of an algorithm.
When the number of processors $n$ is known, then we give Las Vegas algorithms for each of the four cases of naming determined by the kind of PRAM model and the amount of shared memory.
When the number of processors is unknown, then we give Monte Carlo algorithms for each of the respective four  cases of naming.

\Paragraph{The summary of the results.}

We consider randomized algorithms executed by anonymous processors that operate in a synchronous manner using read-write shared memory with the goal to assign unique names to the processors. 
The algorithms have to be randomized (no deterministic exist) and when the number of processors is unknown then they need to be Monte Carlo (no Las Vegas exist).

We show that naming algorithms for $n$ processors using $C>0$ shared memory cells need to operate in $\Omega(n/C)$ expected time on an Arbitrary PRAM, and in $\Omega(n\log n/C)$ expected time on a Common PRAM.
We prove additionally that any naming algorithm needs to work in the expected time $\Omega(\log n)$; this bound is relevant only when there is unbounded supply of shared memory.
We show that, for unknown $n$, a Monte Carlo naming algorithm that uses $\cO(n\log n)$ random bits has to  fail to assign unique names with probability that is~$n^{-\Omega(1)}$.

We consider eight specific naming problems for PRAMs.
They are determined by the following independent specifications: whether $n$ is known or not, what is the amount of shared memory (constant versus  unbounded),  and by the PRAM variant (Arbitrary versus Common).

For the case of known $n$, the naming algorithms we give are all Las Vegas.
The naming problems' specifications and the corresponding algorithms with their performance bounds are summarized in Table~\ref{tab:summary-LV}.
These algorithms operate in asymptotically optimal times, for a given amount of shared memory, and use the optimum expected number $\cO(n\log n)$ of random bits.
When the amount of memory is unbounded, they use only the amount of space that is provably necessary to attain their running-time performance. 

For the case of unknown $n$,  the naming algorithms we give are all Monte Carlo.
The list of the naming problems' specifications and the corresponding algorithms with their performance bounds are summarized in Table~\ref{tab:summary-MC}.
All Monte Carlo algorithms that we give have the polynomial probability of error, which is best possible when using the $\cO(n\log n)$ expected number of random bits.
When the shared memory is bounded, then these algorithms operate in asymptotically optimal times, for bounded memory, and use the optimum expected number $\cO(n\log n)$ of random bits.
When there is unbounded supply of shared memory, then we give two variants of the algorithms for Arbitrary PRAM and two for Common PRAM, with the goal to optimize different performance metrics.
The set of integers used as names always makes a contiguous segment starting from the smallest name~$1$, so that the only possible kind of error is in assigning duplicate names.

\begin{table}[t]
\begin{center}
\begin{tabular}{|l| c |c |l |}
\hline
\RB \LB
PRAM Model&Memory&Time&Algorithm\\
\hline
\hline
\RB \LB
Arbitrary&$\cO(1)$&$\cO(n)$&\textsc{Arbitrary-Bounded-MC} in Section~\ref{sec:arbitrary-bounded-MC}\\
\hline
\RB \LB
Arbitrary&unbounded&polylog&\textsc{Arbitrary-Unbounded-MC} in Section~\ref{sec:arbitrary-unbounded-MC}\\
\hline
\RB \LB
Common&$\cO(1)$&$\cO(n\log n)$&\textsc{Common-Bounded-MC} in Section~\ref{sec:common-bounded-MC}\\
\hline
\RB \LB
Common&unbounded&polylog&\textsc{Common-Unbounded-MC} in Section~\ref{sec:common-unbounded-MC}\\
\hline
\end{tabular}
\FF

\parbox{\pagewidth}{\FF\caption{\label{tab:summary-MC} 
Four naming problems for unknown $n$, as determined by the PRAM model and the available amount of shared memory, with the respective performance bounds. 
All the algorithms are Monte Carlo.
When time is marked as ``polylog'' then this means that the algorithm comes in two variants, such that in one the expected time is $\cO(\log n)$ and the amount of used shared memory is suboptimal $n^{\cO(1)}$, and in the other the expected time is suboptimal $\cO(\log^2 n)$ but the  amount of used shared memory misses optimality only by at most a logarithmic factor.
}}
\end{center}
\end{table}

\Paragraph{Previous and related work.}

The naming problem for a synchronous PRAM has not been previously considered in the literature, to the best of the authors' knowledge.
There is a voluminous literature on various aspects of computing and communication in anonymous systems, the following review is necessarily selective.

We begin with the previous work on naming in shared-memory systems with read-write registers.
A systematic exposition of shared-memory algorithm can be found in~\cite{Attiya-Welch-book2004}, when approached from the distributed-computing perspective, and in~\cite{JaJa92-book}, when approached  from the parallel-computing one.

Lipton and Park~\cite{LiptonP-1990} considered naming in asynchronous distributed systems with read-write shared memory controlled by adaptive schedulers; they proposed a solution that terminates with positive probability, and which can be made arbitrarily close to~$1$ assuming that $n$ is known.
E{\u{g}}ecio{\u{g}}lu and Singh~\cite{EgeciogluS94}  proposed a polynomial-time Las Vegas naming algorithm for asynchronous  systems with known~$n$ and read-write shared memory with oblivious scheduling of events. 
Kutten et al.~\cite{KuttenOP00} provided a thorough study of naming in  asynchronous systems of shared read-write memory.
They gave a Las Vegas algorithm for an oblivious scheduler for the case of known~$n$, which works in the expected time~$\cO(\log n)$ while using $\cO(n)$ shared registers, and also showed that a logarithmic time is required to assign names to anonymous processes.
Additionally, they showed  that if $n$ is unknown then a Las Vegas naming algorithm does not exist,  and a finite-state Las Vegas naming algorithm can work only for an oblivious scheduler.
Panconesi et al.~\cite{PanconesiPTV98} gave a randomized wait-free naming algorithm for  anonymous systems with processes prone to crashes that communicate by single-writer registers. The model considered in that work assigns unique single-writer registers to nameless processes and so has a potential to defy the impossibility of wait-free naming for general multi-writer registers, that impossibility proved by Kutten et al.~\cite{KuttenOP00}.
Buhrman et al.~\cite{BuhrmanPSV-2006} considered the relative complexity of naming and consensus problems in asynchronous systems with shared memory that are prone to crash failures, demonstrating that naming is harder than consensus.

Next, we review work on problems in anonymous distributed systems different from naming.
Aspnes et al.~\cite{AspnesFR06} gave a comparative study of anonymous distributed systems with different communication mechanisms, including broadcast and shared-memory objects of various functionalities, like read-write registers and counters. 
Alistarh et al.~\cite{AlistarhAGGG10} gave randomized renaming algorithms that act like naming ones, in that process identifiers are not referred to; for more or renaming see~\cite{AlistarhACGG14,AttiyaBDPR90,ChlebusK08}.
Aspnes et al.~\cite{AspnesSS02} considered solving consensus in anonymous systems with infinitely many processes. 
Attiya et al.~\cite{AttiyaGM02} and Jayanti and Toueg~\cite{JayantiT90}  studied the impact of  initialization of shared registers on solvability of tasks like consensus and wakeup in fault-free anonymous systems.
Bonnet et al.~\cite{BonnetR-2011} considered solvability of consensus in anonymous systems with processes prone to crashes but augmented with failure detectors.
Guerraoui and Ruppert~\cite{GuerraouiR-2007} showed that certain tasks like time-stamping, snapshots and consensus have deterministic solutions in anonymous systems  with shared read-write registers and with processes  prone to crashes.
Ruppert~\cite{Ruppert2007} studied the impact of anonymity of processes on wait-free computing and mutual implementability of types of shared objects.

The problem of concurrent communication in anonymous networks was first considered by Angluin~\cite{Angluin80}.
That work showed, in particular, that randomization was needed in naming algorithms when executed in environments that are perfectly symmetric; other related impossibility results were surveyed by Fich and Ruppert~\cite{FichR03}.

The work about anonymous networks that followed was either on specific network topologies or on problems in general message-passing systems.
Most popular specific topologies included that of a ring and hypercube.
In particular, the ring topology was investigated by Attiya et al.~\cite{AttiyaS91,AttiyaSW88}, Flocchini et al.~\cite{FlocchiniKKLS04},  Diks et al.~\cite{DiksKMP95}, Itai and Rodeh~\cite{ItaiR90}, and Kranakis et al.~\cite{KranakisKL01}, and the hypercube topology was studied by Kranakis and Krizanc~\cite{KranakisK97} and Kranakis and Santoro~\cite{KranakisS01}.

The work on algorithmic problems in anonymous networks of general topologies or anonymous/named  agents in anonymous/named networks included the following contributions. 
Afek and Matias~\cite{AfekM94} and Schieber and Snir~\cite{SchieberS94} considered leader election, finding spanning trees and naming in general anonymous networks.
Angluin et al.~\cite{AngluinAER05} studied adversarial communication by anonymous agents and Angluin et al.~\cite{AngluinAFJ08} considered self-stabilizing protocols for anonymous asynchronous agents deployed in a network of unknown size.
Chalopin et al.~\cite{ChalopinMM12} studied naming and leader election in asynchronous networks when a node knows the map of the network but its position on the map is unknown.
Chlebus et al.~\cite{ChlebusDT-beeping-2016} considered assigning names to anonymous stations attached to a channel that allows only beeps to be heard.
Chlebus et al.~\cite{ChlebusDP94} investigated anonymous complete networks whose links and nodes are subject to random independent failures in which single fault-free node has to wake up all nodes by propagating a wakeup message through the network.
Dereniowski and Pelc~\cite{DereniowskiP14} considered leader election among anonymous agents in anonymous networks.
Dieudonn{\'{e} and Pec~\cite{DieudonneP16} studied teams of anonymous mobile agents in networks that execute deterministic algorithm with the goal to convene at one node.
Fraigniaud et al.~\cite{FraigniaudPPP01} considered naming in anonymous networks with one node distinguished as leader.
G\k{a}sieniec et al.~\cite{GasieniecKKZ06} investigated anonymous agents pursuing the goal to meet at a node or edge of a ring. 
Glacet et al.~\cite{GlacetMP16} considered leader election in anonymous trees.
Kowalski and Malinowski~\cite{KowalskiM08} studied named agents meeting in anonymous networks.
Kranakis et al.~\cite{KranakisKB94} investigated computing boolean functions on anonymous networks.
M{\'{e}}tivier et al.~\cite{MetivierRZ15} considered naming anonymous unknown graphs. 
Michail et al.~\cite{MichailCS13} studied the problems of naming and counting nodes in dynamic anonymous networks.
Pelc~\cite{Pelc07} considered activating an anonymous ad hoc radio network from a single source  by a deterministic algorithm. 
Yamashita and Kameda~\cite{YamashitaK96} investigated topological properties of anonymous networks that are conducive to have deterministic solutions for representative algorithmic problems.

General questions of computability in anonymous message-passing systems implemented in networks were studied by Boldi and Vigna~\cite{BoldiV01}, Emek et al.~\cite{EmekSW14}, and Sakamoto~\cite{Sakamoto99}.

Lower bounds on PRAM were given by Fich et al.~\cite{FichHRW85}, Cook et al.~\cite{CookDR86}, and Beame~\cite{Beame88},  among others.
A review of lower bounds based on information-theoretic approach is given by Attiya and Ellen~\cite{Attiya-Ellen-book-2014}.
Yao's minimax principle was given by Yao~\cite{Yao77}; the book by Motwani and Raghavan~\cite{MotwaniR95} gives examples of applications.

\section{Technical Preliminaries}

\label{sec:technical-preliminaries}

A distributed  system with shared memory in which some $n$ processors operate concurrently is the model of computation that we use in this paper.
The essential properties of such systems that we assume are, first, that shared memory cells have only reading/writing capabilities, and, second, that operations of accessing  shared registers are globally synchronized so that processors work in lockstep.

An execution of a synchronous algorithm is structured as a sequence of \emph{rounds} so that each processor performs either a read from a shared memory cell or a write to a shared memory cell in a round.
We assume that a processor carries out its private computation in a round in a negligible portion of the round.
An invocation of either reading from or writing to a memory location is completed in the round of invocation.
This model of computation is referred to in the literature as the \emph{Parallel Random Access Machine (PRAM)}; see   \cite{JaJa92-book,Reif-editor-book}.

PRAM  is usually defined as a model with unbounded number of shared-memory cells, by analogy with the random-access machine (RAM) model for sequential computing~\cite{Boas90}.
In this paper, we consider the following two instantiations of the PRAM model, determined by the amount of the available shared memory. 
In one situation, there is a constant number of shared memory cells, which is independent of the number of processors~$n$ but as large as needed in the specific algorithm.
In the other case, the number of shared memory cells is unbounded in principle, but the expected number of shared registers accessed in an execution depends on~$n$ and is sought to be minimized.

Each shared memory cell is assumed to be initialized to~$0$ as a default value.
This assumption simplifies the exposition, but it is not crucial as any algorithm assuming such an initialization can be modified  to work with dirty memory; for example, one can apply an approach  similar to that in~\cite{KuttenOP00}.
A shared memory cell can store any value as needed in algorithms, in particular, integers of magnitude that may depend on~$n$; all our algorithms require a memory cell to store $\cO(\log n)$ bits.
Processors can generate as many private random bits per round as needed; all these random bits generated in an execution are assumed to be independent.

\Paragraph{PRAM variants.}

Two operations are said to be performed concurrently when they are invoked in the same round of an execution of an algorithm on a PRAM.
A \emph{concurrent read} occurs when a group of processors read from the same memory cell in the same round;  this results in each of these processors obtaining the value stored in the memory cell at the end of the preceding round.
A \emph{concurrent write} occurs when a group of processors invoke a write to the same memory cell in the same round.

Without loss of generality, we may assume that a concurrent read of a memory cell and a concurrent write to the same memory cell do not occur simultaneously: this is because we could designate rounds only for reading and only for writing depending on their parity, thereby slowing the algorithm by a factor of two.

The meaning of concurrent reading from the same memory cell is straightforward, in that all the readers get the value stored in this memory cell.
We need to specify which value gets written to a memory cell in a concurrent write, when multiple distinct values are attempted to be written.
Such stipulations determine the corresponding variants of the model.
We will consider algorithms for the following two PRAM variants determined by their respective concurrent-write semantics.
\begin{description}
\item[\rm \em Common PRAM] is defined by the property that when a group of  processors want to write to the same shared memory cell in a round then all the values that any of the processors want to write must be identical, otherwise the operation is illegal.
Concurrent attempts to write the same value to a memory cell result in this value getting written in this round.
\item[\rm \em Arbitrary PRAM] allows attempts to write any legitimate values to the same memory cell in the same round.
When this occurs, then one of these values gets written, while a selection of this written value is arbitrary.
All possible selections of values that get written need to be taken into account when arguing about correctness of an algorithm.
\end{description}

We will rely on certain standard algorithms developed for PRAMs, as explained in~\cite{JaJa92-book,Reif-editor-book}.
One of them is for prefix-type computations.
A typical situation in which it is applied occurs when there is an array of $m$ shared memory cells, each memory cell storing either $0$ or~$1$.
This may represent an array of bins where $1$ stands for a nonempty bin while $0$ for an empty bin.
Let the rank of a nonempty bin of address~$x$ be the number of nonempty bins with addresses smaller than or equal to~$x$.
Ranks can be computed in time $\cO(\log m)$ by using an auxiliary memory of~$\cO(m)$ cells, assuming there is at least one processor assigned to each nonempty bin, while other processors do not participate; such assignment for anonymous processors will be determined by writes they performed to the bins.
The underlying idea is that bins are associated with the leaves of a binary tree.
The processors traverse a binary tree from the leaves to the root and back to the leaves.
When updating information at a node, only  the information stored at the parent, the sibling and the children is used.

We may observe that the same memory can be used repeatedly when such computation needs to be performed multiple times on the same tree.
A possible approach is to verify if the information at a needed memory cell, representing either a parent, a sibling or a child of a visited node, is fresh or rather stale from previous executions.
This could be accomplished in the following three steps by a processor.
First, the processor erases a memory cell it needs to read by rewriting its present value by a blank value. 
Second, the processor writes again the value at the node it is visiting, which may have been erased in the previous step by other processors that need the value.
Finally, the processor reads again the memory cell it has just erased,  to see if it stays erased, which means its contents were stale,  or not, which in turn means its contents got rewritten so they are fresh.

\Paragraph{Balls into bins.}

In the course of probabilistic analysis of algorithms, we will often model actions of processors by  throwing balls into bins.
This can be done in two natural ways.
One is such that memory addresses are interpreted as bins and the values written represent balls, possibly with labels.
Then the total number of balls considered will always be~$n$, that is, equal to the number of processors of  a PRAM.
Another possibility is when bins represent rounds and selecting a bin results in performing a write to a suitable shared register in the respective round. 

The following terms refer to the status of a bin in a given round.
A bin is called \emph{empty} when there are no balls in it.
A ball is \emph{single} in a bin when there are no other balls in the same bin, and such a bin can be called \emph{singleton}.
A bin is \emph{multiple} when there are at least two balls in it.
A \emph{collision} occurs in a multiple bin.
Finally, a bin with at least one ball is \emph{occupied}.

The \emph{rank} of a bin containing a ball is the number of bins with smaller or equal names that contain balls.
When each processor, in a group of processors that still seek names, throws a ball  and there is no collision then this breaks symmetry in a manner that in principle could facilitate assigning unique names to processors in the group, namely, the ranks of selected bins may serve as names.

Throwing balls into bins will be performed repeatedly in each instance of modeling the behavior of an algorithm.
Each instance of throwing a number of balls into bins is then called a \emph{stage}.
There will be an additional numeric parameter $\beta>0$,  and we call the process of throwing balls into bins the \emph{$\beta$-process}, accordingly.
This parameter $\beta$ may determine the number of bins in a stage and also when a stage is the last one in an execution of the $\beta$-process.
The specifications of any such a $\beta$-process apply only within the section in which it is defined.

In a given stage of a $\beta$-process, the balls that are thrown into the bins are called \emph{eligible} for the stage; all the $n$ balls are eligible for the first stage.
When a bin is selected for a ball to be placed in, then this occurs uniformly at random over the range of bins, and independently over the considered balls.
Each selection of a bin for an eligible balls requires the number of random bits equal to the binary logarithm of the number of the available bins.

When we sum up the numbers of available bins over all the stages of an execution of a $\beta$-process  until termination, then the result is \emph{the number of bins ever needed}  in this execution.
Similarly, \emph{the number of bits ever generated} in an execution of a $\beta$-process is the sum of all the numbers of random bits needed to be generated to place balls, over all the stages and balls until termination of this  execution.

The idea of representing attempts to assign names as throwing balls into bins is quite generic.
In particular, it was applied by E{\u{g}}ecio{\u{g}}lu and Singh~\cite{EgeciogluS94}, who proposed a synchronous algorithm that repeatedly throws all balls together into all available bins, the selections of bins for balls made independently and uniformly at random.
In their algorithm for $n$ processors, we can use $\gamma \cdot n$ memory cells, where $\gamma>1$.
Let us choose $\gamma=3$ for the following calculations to be specific.
This algorithm has an exponential expected-time performance.
To see this, we estimate the probability that each bin is either singleton or empty.
Let the balls be thrown one by one.
After the first $n/2$ balls are in singleton bins, the probability to hit an empty bin is at most~$\frac{2.5 n}{ 3 n}=\frac{5}{6}$; we treat this as a success in a Bernoulli trial.
The probability of $n/2$ such successes is at most~$(\frac{5}{6})^{n/2}$, so the expected time to  wait for the algorithm to terminate is at least~$\bigl(\sqrt{\frac{6}{5}}\bigr)^n$, which is exponential in~$n$.

We consider related processes that could be as fast as $\cO(\log n)$ in expected time, while still using only $\cO(n)$ shared memory cells, see Section~\ref{sec:common-unbounded-LV}.
The idea is to let balls in singleton bins stay put and only move those that collided with other balls by landing in  bins that became thereby multiple.
To implement this on a Common PRAM, we need a way to detect collisions, which we explain next.

\Paragraph{Verifying collisions.}

We will use a randomized procedure for Common PRAM to verify if a collision occurs in a bin, say, a bin~$x$, which is executed by each processor that selected bin~$x$.
This procedure \textsc{Verify-Collision} is represented in Figure~\ref{proc:verify-collision}.
There are two arrays \texttt{Tails} and \texttt{Heads} of shared memory cells.
Bin $x$ is verified by using memory cells \texttt{Tails}[$x$] and \texttt{Heads}[$x$]. 
First, the memory cells \texttt{Tails}[$x$] and \texttt{Heads}[$x$] are set  each to false.
Next, each processors selects randomly and independently one of these memory cells and sets it to true.
Finally, every processor reads reads both \texttt{Tails}[$x$] and \texttt{Heads}[$x$] and detects a collision upon reading \texttt{true} twice.


\begin{figure}[t]
\rule{\textwidth}{0.75pt}

\F 
\textbf{Procedure} \textsc{Verify-Collision}\,$(x)$

\rule{\textwidth}{0.75pt}
\begin{center}
\begin{minipage}{\pagewidth}
\begin{description}
\item[\rm initialize] $\texttt{Heads}[x] \gets \texttt{Tails}[x] \gets \texttt{false}$  
\item[\tt toss$_v$] $\gets$  outcome of tossing a fair coin\ 
\item[\tt if] \texttt{toss}$_v$ = tails 
\begin{description}
\item[\texttt{then}] \texttt{Tails}[$x$] $\gets \texttt{true}$  

\item[\texttt{else}] \texttt{Heads}[$x$] $\gets \texttt{true}$
\end{description}
\item[\tt return] (\texttt{Tails}[$x$] = \texttt{Heads}[$x$])
\end{description}
\end{minipage}
\FFF

\rule{\textwidth}{0.75pt}

\parbox{\captionwidth}{\caption{\label{proc:verify-collision}
A pseudocode for a  processor~$v$ of a Common PRAM, where $x$ is a positive integer. 
\texttt{Heads} and \texttt{Tails} are arrays of shared memory cells.
When the parameter $x$ is dropped in a call then this means that $x=1$.
The procedure returns \texttt{true} when a collision is detected.
}}
\end{center}
\end{figure}


\begin{lemma}
\label{lem:verify-collision}

For an integer $x$, procedure \textsc{Verify-Collision}\,$(x)$ executed by one processor never detects a collision, and when multiple processors execute this procedure then a collision is detected with probability at least $\frac{1}{2}$.
\end{lemma}

\begin{proof}
When only one processor executes the procedure, then first the processor sets both \texttt{Heads}[$x$] and \texttt{Tails}[$x$] to false and next only one of them to true.
This guarantees that \texttt{Heads}[$x$] and \texttt{Tails}[$x$] store different values and so collision is not detected.
When some $m>1$ processors execute the procedure, then collision is not detected only when either all processors set \texttt{Heads}[$x$] to true or all processors set \texttt{Tails}[$x$] to true.
This means that the processors generate the same outcome in their coin tosses. 
This occurs with probability $2^{-m+1}$, which is at most $\frac{1}{2}$.
\end{proof}

\Paragraph{Pseudocode conventions and notations.}

We specify algorithms using pseudocode conventions natural for the PRAM model.
An example of such a representation is in Figure~\ref{proc:verify-collision}.
These conventions are summarized as follows.

There are two kinds of variables: shared and private.
The names of shared variables start with capital letters and the names of private ones are all in small letters.
To emphasize that a private variable $x$ is such that its value may depend on a processor~$v$ in a round, we may denote $x$ by~$x_v$ in a pseudocode for~$v$.

When~$x$ is a private variable that may have different values at different processors at the same time, then  we denote this variable used by a processor~$v$ by~$x_v$.
Private variables that have the same value at the same time in all the processors are usually used without subscripts, like variables controlling for-loops.
An assignment instruction of the form $x\gets y\gets\ldots\gets z\gets \alpha$, where $x,y,\dots, z$ are variables and $\alpha$ is a value, means to assign~$\alpha$ as the value to be stored in all the listed variables~$x,y,\dots, z$.

We want that, at any round of an execution, all the processors that have not terminated yet are executing the same line of the pseudocode.
In particular, when an instruction is conditional on a statement then a processor that does not meet the condition pauses as long as it would be needed for all the processors that meet the condition complete their instructions, even when there are no such processors.
If this is not a constant-time instruction, then it may incur an unnecessary time cost.
To avoid this problem, we may first verify if there is some processor that satisfies the condition, which can be done in constant time.

We use three notations for logarithms.
The notation $\lg x$ stands for the logarithm of $x$ to the base~$2$.
The notation $\ln x$ denotes the natural logarithm of~$x$.
When the base of logarithms does not matter then we use~$\log x$, like in the asymptotic notation~$\cO(\log x)$.

\Paragraph{Properties of  naming algorithms.}

Naming algorithms in distributed environments involving multi-writer read-write shared memory have to be randomized to break symmetry~\cite{Angluin80,Attiya-Welch-book2004}.
An eventual assignment of proper names cannot be a sure event, because, in principle, two processors can generate the same strings of random bits in the course of an execution.
We say that an event \emph{is almost sure}, or \emph{occurs almost surely}, when it occurs with probability~$1$. 
When $n$ processors generate their private strings of random bits then it is an almost sure event that all these strings are eventually pairwise distinct.
Therefore, a most advantageous scenario that we could expect, when a set of $n$ processors is to execute a randomized naming algorithm, is that the algorithm eventually terminates almost surely and that at the moment of termination the output is \emph{correct}, in that the assigned names are without duplicates and fill the whole interval $[1,n]$.

Randomized naming algorithms are categorized as either Monte Carlo or Las Vegas, which are defined as follows.
A randomized algorithm is \emph{Las Vegas} when it terminates almost surely and the algorithm returns a correct output upon termination.
A randomized algorithm is \emph{Monte Carlo} when it terminates almost surely and an incorrect output may be produced upon termination, but the probability of error converges to zero with the size of input growing unbounded.

We give algorithms that use the expected number of $\cO(n \log n)$ random bits with a large probability.
This amount of random information is necessary if an algorithm is to terminate almost surely.
The following fact is essentially a folklore, but since we do not know if it was ever proved  in the literature, we give a proof for completeness' sake.
Our arguments resort to the notions of information theory~\cite{CoverT-book}.


\begin{proposition}
\label{pro:lower-bound-on-random-bits}

If a randomized naming algorithm is correct with probability~$p_n$, when executed by $n$ anonymous processors, then it requires $\Omega(n\log n)$ random bits with probability at least~$p_n$.
In particular, a Las Vegas naming algorithm for $n$ processors uses $\Omega(n\log n)$ random bits almost surely.
\end{proposition}

\begin{proof}
Let us assign conceptual identifiers to the processors, for the sake of argument.
These \emph{unknown identifiers} are known only to an external observer and not to algorithms.
The purpose of executing the algorithm is to assign explicit identifiers, which we call  \emph{given names}.

Let a processor with an unknown identifier~$u_i$ generate a string of bits~$b_i$, for $i=1,\ldots, n$.
A distribution of given names among the $n$ anonymous processors, which results from executing the algorithm, is a random variable $X_n$ with a uniform distribution on the set of all permutations of the unknown identifiers.
This is because of symmetry: all processors execute the same code, without explicit private names, and if we rearrange the strings generated bits~$b_i$ among the processors~$u_i$, then this results in the corresponding rearrangement of the given names.

The underlying probability space consists of $n!$ elementary events, each determined by an assignment of the given names to the processors identified by the unknown identifiers.
It follows that each of these events occurs with probability~$1/n!$.
The Shannon entropy of the random variable~$X_n$ is thus $\lg (n!)=\Theta(n\log n)$.
The decision about which assignment of given names is produced is determined by the random bits, as they are the only source of entropy.
It follows that the expected number of random bits used by the algorithm needs to be as large as the entropy of the random variable~$X_n$.

The property that all assigned names are distinct and in the interval $[1,n]$ holds with probability~$p_n$.
An execution needs to generate a total of $\Omega(n\log n)$ random bits with probability at least~$p_n$, because of the bound on entropy.
A Las Vegas algorithm terminates almost surely, and returns correct names upon termination.
This means that  $p_n=1$ and so that $\Omega(n\log n)$ random bits are used almost surely.
\end{proof}

A naming algorithm cannot be Las Vegas when $n$ is unknown, as was observed by Kutten et al.~\cite{KuttenOP00} for asynchronous computations against an oblivious adversary.
We show the analogous fact for synchronous computations.


\begin{proposition} 
\label{pro:n}

There is no Las Vegas naming algorithm for a PRAM with $n>1$ processors that does not refer to the number of processors~$n$ in its code.
\end{proposition}

\begin{proof}
Let us suppose, to arrive at a contradiction, that such a naming Las Vegas algorithm exists.
Consider a system of $n -1\ge  1$ processors, and an execution~$\cE$ on these $n-1$ processors that uses specific strings of random bits such that the algorithm terminates in~$\cE$ with these random bits.
Such strings of random bits exist because the algorithm terminates almost surely.

Let $v_1$ be a processor that halts latest in~$\cE$ among the $n-1$ processors.
Let $\alpha_{\cE}$ be the string of random bits generated by processor~$v_1$ by the time it halts in~$\cE$.
Consider the execution~$\cE'$ on $n\ge 2$ processors such that $n-1$ processors obtain the same strings of random bits as in~$\cE$ and an extra processor~$v_2$ obtains~$\alpha_{\cE}$ as its random bits.
The executions  $\cE$ and~$\cE'$ are indistinguishable for the~$n-1$ processors participating in~$\cE$, so they assign themselves the same names and halt.
Processor~$v_2$ performs the same reads and writes as processor~$v_1$ and assigns itself the same name as processor~$v_1$ does and halts in the same round as processor $v_1$.
This is the termination round because by that time all the other processor have halted as well.

It follows that execution $\cE'$ results in a name being duplicated.
The probability of duplication for $n$ processors is at least as large as the probability to generate two identical finite random strings in $\cE'$ for some two  processors, so this probability is positive.
\end{proof}

If $n$ is unknown, then the restriction $\cO(n \log n)$ on the number of random bits makes it  inevitable that the probability of error is  at least polynomially bounded from below, as we show next.


\begin{proposition}
\label{pro:probability-of-error}

For unknown $n$, if a randomized naming algorithm is executed by $n$ anonymous processors, then an execution is incorrect, in that duplicate names are assigned to distinct processors, with probability that is at least~$n^{-\Omega(1)}$, assuming that the algorithm uses $\cO(n\log n)$ random bits with probability $1-n^{-\Omega(1)}$.
\end{proposition}

\begin{proof}
Suppose the algorithm uses at most $c n\lg n$ random bits with  probability~$p_n$ when executed by a system of $n$ processors, for some constant $c>0$.
Then one of these processors uses at most $c \lg n$ bits with  probability~$p_n$, by the pigeonhole principle.

Consider an execution for $n+1$ processors.
Let us distinguish a processor~$v$.
Consider the actions of the remaining $n$ processors: one of them, say~$w$, uses at most $c \lg n$ bits with the probability~$p_n$. 
Processor $v$ generates the same string of bits with probability~$2^{-c\lg n}= n^{-c}$.
The random bits generated by $w$ and $v$ are independent.
Therefore duplicate names occur with probability at  least~$n^{-c}\cdot p_n$.
When we have a bound on probabilities~$p_n$ to be $p_n=1-n^{-\Omega(1)}$, then probability of occurrence of duplicate names is at least $n^{-c}(1-n^{-\Omega(1)})= n^{-\Omega(1)}$.
\end{proof}

\section{Lower Bounds on Running Time}

\label{sec:lower-bounds-on-time}

We consider two kinds of algorithmic naming problems, as determined by the amount of shared memory.
One case is for a constant number of shared memory cells, for which we give an optimal lower bound on time for $\cO(1)$ shared memory.
The other case is when the number of shared memory cells and their capacity are unbounded, for which we give an ``absolute'' lower bound on time.
We begin with lower bounds that reflect the amount of shared memory.

Intuitively, as processors generate random bits, these bits need to be made common knowledge through some implicit process that assigns explicit names.
There is an underlying flow of information to spread knowledge among the processors through the available shared memory.
Time is bounded from below by the rate of flow of information and the total amount of bits that need to be shared.

On the technical level, in order to bound the expected time of a randomized algorithm, we apply the Yao's minimax principle~\cite{Yao77} to relate this expected time to  the distributional expected time complexity.
A randomized algorithm whose actions are determined by random bits can be considered as a  probability distribution on deterministic algorithms.
A deterministic algorithm has strings of bits given to processors as their inputs, with some probability distribution on such inputs.
The expected time of such a deterministic algorithm, give any specific probability distribution  on the inputs, is a lower bound on the expected time of a randomized algorithm.

To make such interpretation of randomized algorithms possible, we consider strings of bits of equal length. 
With such a restriction on inputs, deterministic algorithm may not be able to assign proper names for some assignments of inputs, for example, when all the inputs are equal.
We augment such deterministic algorithms by adding an option for the algorithm to withhold a decision on assignment of names and output ``no name'' for some processors.
This is interpreted as the deterministic algorithm needing longer inputs, for which the given inputs are prefixes, and which for the randomized algorithm means that some processors need to generate more random bits. 

Regarding probability distributions for inputs of a given length, it will always be the uniform distribution.
This is because we will use an assessment of the amount of entropy of such a distribution.


\begin{theorem}
\label{thm:lower-bound-memory-common-pram}

A randomized naming algorithm for a Common PRAM with $n$ processors and $C>0$ shared memory cells operates in $\Omega(n\log n/C)$ expected time when it is either a Las Vegas algorithm or a Monte Carlo algorithm with the probability of error smaller than~$1/2$.
\end{theorem}

\begin{proof}
We consider Las Vegas algorithms in this argument, the Monte Carlo case is similar, the difference is in applying Yao's principle for Monte Carlo algorithms.
We interpret a randomized algorithm as a deterministic one working with all possible assignments of random bits as inputs with a uniform mass function on the inputs.
The expected time of the deterministic algorithm is a lower bound on the expected time of the randomized algorithm.

There are $n!$ possible assignments of given names to the processors.
Each of them occurs with the same probability $1/n!$ when the input bit strings are assigned uniformly at random.
Therefore the entropy of name assignments, interpreted as a random variable, is $\lg n!=\Omega(n\log n)$.

Next we consider executions of such a deterministic algorithm on the inputs with a uniform probability distribution.
We may assume without loss of generality that an execution is structured into the following phases, each consisting of $C+1$ rounds.
In the first round of a phase, each processor either writes into a shared memory cell or pauses.
In the following  rounds of a phase, every processor learns the current values of each among the $C$ memory cells.
This may take $C$ rounds for every processor to scan the whole shared memory, but we do not include this reading overhead as contributing to the lower bound.
Instead, since this is a simulation anyway, we conservatively assume that the process of learning  all the contents of shared memory cells at the end of a phase is instantaneous and complete.

The Common variant of PRAM requires that if a memory cell is written into concurrently then there is a common value that gets written by all the writers.
Such a value needs to be determined by the code and the address of a memory cell.
This means that, for each phase and any memory cell, a  processor choosing to write into this memory cell knows the common value to be written.
By the structure of execution, in which all processors read all the registers after a round of writing, any processor knows what value gets written into each available memory cell in a phase, if any is written into a particular cell.
This implies that the contents written into shared memory cells may not convey any new information but are already implicit in the states of the processors represented by their private memories after reading the whole shared memory.

When a processor reads all the shared memory cells in a phase, then the only new information it may learn is the addresses of memory cells into which new writes were performed and those into which there were no new writes.
This makes it possible obtain at most $C$ bits of information per phase, because each register  was either written into or not.

There are $\Omega(n\log n)$ bits of information that need to be settled and one phase changes the entropy by at most $C$ bits.
It follows that the expected number of phases of the deterministic algorithm is $\Omega(n\log n /C)$.
By the Yao's principle, $\Omega(n\log n/C)$ is a lower bound on the expected time of a randomized algorithm.
\end{proof}

For Arbitrary PRAM, writing can spread information through the written values, because different processes can attempt to write distinct strings of bits.
The rate of flow of information is constrained by the fact that when multiple writers attempt to write to the same memory cell then only one of them succeeds, if the values written are distinct.
This intuitively means that the size of a group of processors  writing to the same register determines how much information the writers learn by subsequent reading.
These intuitions are made formal in the proof of the following Theorem~\ref{thm:lower-bound-memory-arbitrary-pram}.


\begin{theorem}
\label{thm:lower-bound-memory-arbitrary-pram}

A randomized naming algorithm for an Arbitrary PRAM with $n$ processors and $C>0$ shared memory cells operates in $\Omega(n/C)$ expected time when it is either a Las Vegas algorithm or a Monte Carlo algorithm with the probability of error smaller than~$1/2$.
\end{theorem}

\begin{proof}
We consider Las Vegas algorithms in this argument, the Monte Carlo case is similar, the difference is in applying Yao's principle for Monte Carlo algorithms.
We again replace a given randomized algorithm by its deterministic version that works on assignments of strings of bits of the same length as inputs, with such inputs assigned uniformly at random to the processors. 
The goal is to use the property that the expected time of this deterministic algorithm, for a given probability distribution of inputs, is a lower bound on the expected time of the randomized algorithm.
Next, we consider executions of this deterministic algorithm.

Similarly as in the proof of Theorem~\ref{thm:lower-bound-memory-common-pram}, we observe that  there are $n!$  assignments of given names to the processors and each of them occurs with the same probability $1/n!$, when the input bit strings are assigned uniformly at random.
The entropy of name assignments is again $\lg n!=\Omega(n\log n)$.
The algorithm needs to make the processors learn $\Omega(n\log n)$ bits using the available $C>0$ shared memory cells.

We may interpret an execution as structured into phases, such that each processor performs at most one write in a phase and then reads all the registers.
The time of a phase is assumed conservatively to be $\cO(1)$.
Consider a register and a group of processors that attempt to write their values into this register in a phase.
The values attempted to be written are represented as strings of bits.
If some of these values have~$0$ and some have~$1$ at some bit position among the strings, then this bit position may convey one bit of information.
The maximum amount of information is provided by a write when the written string of bits facilitates identifying the writer by comparing its written value to the other values attempted to be written concurrently to the same memory cell.
This amount is at most the binary logarithm of the size of this group of processors.
Therefore,  each memory cell written to in a round contributes at most $\lg n$ bits of information, because there may be at most $n$ writers to it.
Since there are $C$ registers, the maximum number of bits of information learnt by the processors in a phase is $C\lg n$.

The entropy of the  assignment of names is $\lg n!=\Omega(n\log n)$, so the expected number of phases of the deterministic algorithm is $\Omega(n\lg n /(C\lg n))=\Omega(n/C)$.
By the Yao's principle, this is also a lower bound on the expected time of a randomized algorithm.
\end{proof}

Next, we consider ``absolute'' requirements on time  for a PRAM to assign unique names to the  available $n$ processors.
The generality of the lower bound we give stems from the weakness of the assumptions.
First, nothing is assumed about the knowledge of~$n$.
Second, concurrent writing is not constrained in any way. 
Third, shared memory cells are unbounded in their number and size.

We show next in Theorem~\ref{thm:log-n-lower-bound} that any Las Vegas naming algorithm has $\Omega(\log n)$  expected time for the synchronous schedule of events.
The argument we give is in the spirit of similar arguments applied by Cook et al.~\cite{CookDR86} and  Beame~\cite{Beame88}.
In an analogous manner, Kutten et al.~\cite{KuttenOP00} showed that any Las Vegas naming algorithm for asynchronous read-write shared memory systems has the expected time $\Omega(\log n)$ against a certain oblivious schedule.
What these arguments share, along with the arguments we employ in this paper,  are a formalization of the notion of flow of information during an execution of an algorithm, combined with a recursive estimate of the rate of this flow.

The relation \emph{processor~$v$ knows processor~$w$ in round~$t$} is defined recursively as follows.
First, for any processor~$v$, we have that~$v$ knows~$v$ in any round~$t>0$.
Second, if a processor~$v$ writes to a shared memory cell~$R$ in a round $t_1$ and a processor~$w$ reads from~$R$ in a round $t_2>t_1$, such that there was no other write into this memory cell after $t_1$ and prior to~$t_2$, then processor~$w$ knows in round~$t_2$ each processor that $v$ knows in round~$t_1$.
Finally, the relation is the smallest transitive relation that satisfies the two postulates formulated above.
This means that it is the smallest relation such that if processor~$v$ knows processor~$w$ in round $t_1$ and $z$ knows $v$ in round $t_2$ such that $t_2>t_1$ then processor~$z$ knows~$w$ in round~$t_2$.
In particular, the knowledge accumulates with time, in that if a processor~$v$ knows processor~$z$ in round~$t_1$ and round~$t_2$ is such that $t_2>t_1$ then $v$ knows~$z$ in round~$t_2$ as well.


\begin{lemma}
\label{lem:knowledge}

Let $\cA$ be a deterministic algorithm that assigns distinct names to the processors, with the possibility that some processors output ``no name'' for some inputs, when each node has an input string of bits of the same length.
When algorithm $\cA$ terminates with proper names assigned to all the processors then each processor knows all the other processors.
\end{lemma}

\begin{proof}
We may assume that $n>1$ as otherwise one processors knows itself.
Let us consider an assignment~$\cI$ of inputs  that results in a proper assignment of distinct names to  all the processors when algorithm~$\cA$ terminates.
This implies that all the inputs in the assignment~$\cI$ are distinct strings of bits, as otherwise some two processors, say, $v$ and $w$ that obtain the same input string of bits would either assign themselves the same name or declare ``no name'' as output.

Suppose that a processor~$v$ does not know a processor~$w$, when $v$ halts for inputs from~$\cI$.
Consider an assignment of inputs $\cJ$ which is the same as $\cI$ for processors different from $w$ and such that the input of $w$ is the same as input for $v$ in~$\cI$.
Then the actions of processor $v$ would be the same with $\cJ$ as with $\cI$, because $v$ is not affected by the input of~$w$, so that  $v$ would assign itself the same name with $\cJ$ as with $\cI$.
But the actions of processor $w$ would be the same in $\cJ$ as those of~$v$, because their input strings of bits are identical under~$\cJ$.
It follows that $w$ would assign itself the name of~$v$, resulting in duplicate names.
This contradicts the assumption that all processors obtain unique names in the execution.
\end{proof}

We will use Lemma~\ref{lem:knowledge} to asses running times by estimating the number of interleaved reads and writes needed for processors to get to know all the processors. 
The rate of learning such information may depend on time, because we do not restrict the amount of shared memory, unlike in Theorems~\ref{thm:lower-bound-memory-common-pram} and~\ref{thm:lower-bound-memory-arbitrary-pram}.
Indeed, the rate may increase exponentially, under most liberal estimates.

The following Theorem~\ref{thm:log-n-lower-bound} holds for both Common and Arbitrary PRAMs. The argument used in the proof is general enough not to depend on any specific semantics of writing.


\begin{theorem}
\label{thm:log-n-lower-bound}

A randomized naming algorithm for a PRAM with $n$ processors  operates in $\Omega(\log n)$ expected time when it is either a Las Vegas algorithm or a Monte Carlo algorithm with the probability of error smaller than~$1/2$.
\end{theorem}

\begin{proof}
The argument is for a Las Vegas algorithm, the Monte Carlo case is similar.
A randomized algorithm can be interpreted as a probability distribution on a finite set of deterministic algorithms.
Such an interpretation works when input strings for a deterministic algorithm are of the same  length.
We consider all such possible lengths for deterministic algorithms, similarly as in the previous proofs of lower bounds.

Let us consider a deterministic algorithm $\cA$, and let inputs be strings of bits of the same length.
We may structure an execution of this algorithm $\cA$ into \emph{phases} as follows.
A phase consists of two rounds. 
In the first round of a phase, each processor either writes to a shared memory cell or pauses.
In the second round of a phase, each processor either reads from a shared memory cell or pauses.
Such structuring can be done without loss of generality at the expense of slowing down an execution by a factor of at most~$2$. 
Observe that the knowledge in the first round of a phase is the same as in the last round of the preceding phase.

Phases are numbered by consecutively increasing  integers, starting  from~$1$.
A phase $i$ comprised pairs of rounds $\{2i-1, 2i\}$, for integers $i\ge 1$.
In particular, the first phase consists of rounds~$1$ and~$2$.
We also add phase~$0$ that represents the knowledge before any reads or writes were performed.

We show the following invariant, for $i\ge 0$: a processor knows at most~$2^{i}$ processors at the end of phase~$i$.
The proof of this invariant is by induction on~$i$.

The base case is for $i=0$.
The invariant follows from the fact that a processor knows only one processor in phase~$0$, namely itself, and $2^0=1$.

To show the inductive step, suppose the invariant holds for a phase $i\ge 0$ and consider the next phase~$i+1$.
A processor~$v$ may increase its knowledge by reading in the second round of phase~$i+1$.
Suppose the read is from a shared memory cell~$R$.
The latest write into this memory cell occurred by the first round of phase~$i+1$.
This means that the processor~$w$ that wrote to~$R$ by phase~$i+1$, as the last one that did write, knew at most~$2^i$ processors in the round of writing, by the inductive assumption and the fact that what is written in phase $i+1$ was learnt by the immediately preceding phase~$i$.
Moreover, by the semantics of writing, the value written to~$R$ by~$w$ in that round removed any previous information stored in~$R$.
Processor $v$ starts phase $i+1$ knowing at most~$2^i$ processors, and also learns of at most $2^i$ other processors by  reading in phase~$i+1$, namely, those values known by the latest writer of the read contents.
It follows that processor $v$ knows at most $2^i + 2^i = 2^{i+1}$ processors by the end of phase~$i+1$.

When proper names are assigned by such a deterministic algorithm, then each processor knows every other processor, by Lemma~\ref{lem:knowledge}.
A processor knows every other processor in a phase~$j$ such that $2^j\ge n$,  by the invariant just proved.
Such a phase number $j$ satisfies $j\ge \lg n$, and it takes $2\lg n$ rounds to complete $\lg n$ phases.

Let us consider inputs strings of bits assigned to processors uniformly at random. 
We need to estimate the expected running time of an algorithm $\cA$ on such inputs.
Let us observe that, in the context of interpreting deterministic executions for the sake to apply Yao's principle, terminating executions of~$\cA$ that do not result in names assigned to all the processors  could be pruned from a bound on their expected running time, because such executions are determined by bounded input strings of bits that a randomized algorithm would extend to make them sufficiently long to assign proper names.
In other words, from the perspective of randomized algorithms, such prematurely ending executions  do not represent  real terminating ones.

The expected time of $\cA$, conditional on terminating with proper names assigned, is therefore at least~$2\lg n$.
We conclude, by the Yao's principle, that any randomized naming algorithm has $\Omega(\log n)$ expected runtime.
\end{proof}

The three lower bounds on time given in this Section may be applied in two ways.
One is to infer optimality of time for a given amount of shared memory used.
Another is to infer optimality of shared memory use given a time performance.
This is summarized in the following Corollary~\ref{cor:four-optima}.

\begin{corollary}
\label{cor:four-optima}

If the expected time of a naming Las Vegas algorithm is $\cO(n)$ on an Arbitrary PRAM with $\cO(1)$ shared memory, then this time performance is asymptotically optimal.
If the expected time of a naming Las Vegas algorithm is $\cO(n\log n)$ on a Common PRAM with $\cO(1)$ shared memory, then this time performance is asymptotically optimal.
If a Las Vegas naming algorithm operates in time $\cO(\log n)$ on an Arbitrary PRAM using $\cO(n/\log n)$ shared memory cells, then this amount of shared memory is asymptotically optimal.
If a Las Vegas naming algorithm operates in time $\cO(\log n)$ on a Common PRAM using $\cO(n)$ shared memory cells, then this amount of shared memory is optimal.
\end{corollary}

\begin{proof}
We verify that the lower bounds match the assumed upper bounds.
By Theorem~\ref{thm:lower-bound-memory-arbitrary-pram}, a Las Vegas algorithm operates almost surely in $\Omega(n)$ time on an Arbitrary PRAM when space is $\cO(1)$.
By Theorem~\ref{thm:lower-bound-memory-common-pram}, a Las Vegas algorithm operates almost surely in $\Omega(n\log n)$ time on a Common PRAM when space is $\cO(1)$.
By Theorem~\ref{thm:lower-bound-memory-arbitrary-pram}, a Las Vegas algorithm operates almost surely in $\Omega(\log n)$ time on an Arbitrary PRAM when space is $\cO(n/\log n)$.
By Theorem~\ref{thm:lower-bound-memory-common-pram}, a Las Vegas algorithm operates almost surely in $\Omega(\log n)$ time on a Common PRAM when space is $\cO(n)$.
\end{proof}

\section{Las Vegas for Arbitrary with Bounded Memory}

\label{sec:arbitrary-bounded-LV}


\begin{figure}[t]
\rule{\textwidth}{0.75pt}

\F 
\textbf{Algorithm} \textsc{Arbitrary-Bounded-LV}

\rule{\textwidth}{0.75pt}
\begin{center}
\begin{minipage}{\pagewidth}
\begin{description}
\item[\tt repeat]\ 
\B
\begin{description}
\item[\rm initialize] $\texttt{Counter} \gets \texttt{name}_v \gets 0$
\item[$\texttt{bin}_v \gets$] random integer in $[1,n^\beta]$
\item[\tt for] $i\gets 1$ \texttt{to} $n$ \texttt{do} 

\begin{description}
\item[\texttt{if}] $\texttt{name}_v = 0$ \texttt{then}
\begin{description}

\item[\tt Pad] $\gets \texttt{bin}_v$
\item[\tt if] $\texttt{Pad}=\texttt{bin}_v$  \texttt{then}
\begin{description}
\item[\tt Counter]$\gets \texttt{Counter}+1$  
\item[\tt name$_v$]$\gets \texttt{Counter}$
\end{description}
\end{description}
\end{description}
\end{description}
\B

\item[\tt until] $\texttt{Counter} = n$ 
\end{description}
\end{minipage}
\FFF

\rule{\textwidth}{0.75pt}

\parbox{\captionwidth}{\caption{\label{alg:arbitrary-bounded-LV}
A pseudocode for a processor~$v$ of an Arbitrary PRAM, where  the number of shared memory cells is a constant independent of~$n$.
The variables \texttt{Counter} and \texttt{Pad} are shared.
The private variable \texttt{name} stores the acquired name.
The constant $\beta>0$ is parameter to be determined by analysis.
}}
\end{center}
\end{figure}

We present a Las Vegas naming algorithm for an Arbitrary PRAM with a constant number of  shared memory cells, in the case when the number of processors~$n$ is known.

During an execution of this algorithm, processors repeatedly write random strings of bits representing integers to a shared memory cell called  \texttt{Pad}, and next read \texttt{Pad} to verify the outcome of writing.
A processor~$v$ that reads the same value as it attempted to write increments the integer stored in a shared register \texttt{Counter} and uses the obtained number as a tentative name, which it stores  in a private variable \texttt{name}$_v$.
The values of \texttt{Counter} could get incremented a total of less than $n$ times, which occurs when some two processors chose the same random integer to write to the register~\texttt{Pad}.
The correctness of the assigned names is verified by the equality $\texttt{Counter}= n$, because $\texttt{Counter}$ was initialized to zero.
When such a verification fails then this results in another iteration of a series of writes to register~\texttt{Pad}, otherwise the execution terminates and the value stored at \texttt{name}$_v$ becomes the final name of processor~$v$.

This algorithm is called \textsc{Arbitrary-Bounded-LV} and its pseudocode is given in Figure~\ref{alg:arbitrary-bounded-LV}.
The pseudocode refers to a constant $\beta>0$ which determines the bounded range $[1,n^\beta]$ from which processors select integers to write  to the shared register~\texttt{Pad}. 
 
\Paragraph{Balls into bins.}

The selection of random integers in the range $[1,n^\beta]$ by $n$ processors can be interpreted as throwing $n$ balls into $n^\beta$ bins, which we call \emph{$\beta$-process}.
A collision represents two processors assigning themselves the same name.
Therefore an execution of the algorithm can be interpreted as performing such ball placements repeatedly until there is no collision.


\begin{lemma}
\label{lem:balls-arbitrary-bounded}

For each $a>0$ there exists $\beta>0$ such that when $n$ balls are thrown into $n^\beta$ bins during the $\beta$-process then the probability of a collision is at most $n^{-a}$.
\end{lemma}

\begin{proof}
Consider the balls thrown one by one.
When a ball is thrown, then at most $n$ bins are already occupied, so the probability of the ball ending in an occupied bin is at most $n/n^\beta=n^{-\beta+1}$.
No collisions occur with probability that is at least
\begin{equation}
\label{eqn:probability-collision}
\Bigl(1-\frac{1}{n^{\beta-1}}\Bigr)^n\ge 1-\frac{n}{n^{\beta-1}}= 1-n^{-\beta+2}
\ ,
\end{equation}
by the Bernoulli's inequality.
If we take $\beta\ge a+2$ then just one iteration of the repeat-loop is sufficient with  probability that is at least $1-n^{-a}$.
\end{proof}

Next we summarize the performance of algorithm \textsc{Arbitrary-Bounded-LV} as a Las Vegas algorithm.


\begin{theorem}
\label{thm:arbitrary-bounded-LV}

Algorithm \textsc{Arbitrary-Bounded-LV} terminates almost surely and there is no error when it terminates.
For any $a>0$, there exist $\beta>0$ and $c>0$ and such that the algorithm terminates within  time $c n$ using at most $cn\ln n$ random bits with probability at least $1-n^{-a}$.
\end{theorem}

\begin{proof}
The algorithm assigns consecutive names from a continuous interval starting from~$1$,  by the pseudocode in Figure~\ref{alg:arbitrary-bounded-LV}.
It terminates after $n$ different tentative names have been assigned, by the condition controlling the repeat loop in the pseudocode of Figure~\ref{alg:arbitrary-bounded-LV}.
This means that proper names have been assigned when the algorithm terminates.

We map an execution of the $\beta$-process on an execution of algorithm \textsc{Arbitrary-Bounded-LV} in a natural manner.
Under such an interpretation, Lemma~\ref{lem:balls-arbitrary-bounded} estimates the probability of the event that the $n$ processors select different numbers in the interval $[1,n^\beta]$ as their values to write to \texttt{Pad} in one iteration of the repeat-loop.
This implies that just one iteration of the repeat-loop is sufficient with the probability that is at least $1-n^{-a}$.
The probability of the event that $i$ iterations are not sufficient to terminate is at most $n^{-ia}$, which converges to $0$ as~$i$ increases, so  the algorithm terminates almost surely.
One iteration of the repeat-loop takes $\cO(n)$ rounds and it requires $\cO(n\log n)$ random bits.
\end{proof}

Algorithm \textsc{Arbitrary-Bounded-LV} is optimal among Las Vegas naming algorithms with respect to its expected running time~$\cO(n)$, given the amount $\cO(1)$ of its available shared memory, by Corollary~\ref{cor:four-optima} in Section~\ref{sec:lower-bounds-on-time}, and the expected number of random bits $\cO(n\log n)$, by Proposition~\ref{pro:lower-bound-on-random-bits} in Section~\ref{sec:technical-preliminaries}.

\section{Las Vegas for Arbitrary with Unbounded Memory}

\label{sec:arbitrary-unbounded-LV}

In this section, we give a Las Vegas algorithm for an Arbitrary PRAM with an unbounded supply of  shared memory cells, in the case when the number of processors~$n$ is known.
This algorithm is called \textsc{Arbitrary-Unbounded-LV} and its pseudocode  is given in Figure~\ref{alg:arbitrary-unbounded-LV}.


\begin{figure}[t]
\rule{\textwidth}{0.75pt}

\F 
\textbf{Algorithm} \textsc{Arbitrary-Unbounded-LV}

\rule{\textwidth}{0.75pt}
\begin{center}
\begin{minipage}{\pagewidth}
\begin{description}
\item[\tt repeat] \ 
\B
\begin{description}
\item[\rm allocate] \texttt{Counter}\,$[1,\frac{n}{\ln n}]$
\hfill /$\ast$ array of fresh memory cells initialized to $0$s  $\ast$/
\item[\rm initialize] $\texttt{position}_v \gets (0,0)$
\item[$\texttt{bin}_v \gets$] a random integer in $[1,\frac{n}{\ln n}]$
\item[$\texttt{label}_v \gets$] a random integer in $[1,n^\beta]$
\item[\tt repeat]  \ 
\begin{description}
\item[\rm initialize] $\texttt{All-Named}\gets \texttt{true}$
\item[\texttt{if}] $\texttt{position}_v= (0,0)$ \texttt{then}
\begin{description}
\item[\texttt{Bin}\!]$[\texttt{bin}_v]\gets \texttt{label}_v$
\item[\tt if] $\texttt{Bin}\,[\texttt{bin}_v] = \texttt{label}_v$ \texttt{then}
\begin{description}
\item[\texttt{Counter}\!\!]$[\texttt{bin}_v]\gets \texttt{Counter}\,[\texttt{bin}_v]+1$ 
\item[\texttt{position}$_v$]$ \gets (\texttt{bin}_v,\texttt{Counter}\,[\texttt{bin}_v])$
\end{description}
\item[\tt else] \texttt{All-Named} $\gets$ \texttt{false}
\end{description}
\end{description}
\item[\tt until] \texttt{All-Named}
\hfill /$\ast$ each processor has a tentative name $\ast$/
\item[\texttt{name}$_v$]$\gets \text{rank of\ } \texttt{position}_v$
\B
\end{description}
\item[\tt until] $n$ is the maximum name   
\hfill /$\ast$ no duplicates among tentative names $\ast$/
\end{description}
\end{minipage}
\FFF

\rule{\textwidth}{0.75pt}

\parbox{\captionwidth}{\caption{\label{alg:arbitrary-unbounded-LV}
A pseudocode for a processor~$v$ of an Arbitrary PRAM, where the number of shared memory cells is unbounded.
The variables  \texttt{Bin} and \texttt{Counter} denote arrays of $\frac{n}{\ln n}$ shared memory cells each, the variable \texttt{All-Named} is also shared.
The private variable \texttt{name} stores the acquired name.
The constant $\beta>0$ is a parameter to be determined by analysis.
}}
\end{center}
\end{figure}

The algorithm uses two arrays  \texttt{Bin} and \texttt{Counter} of $\frac{n}{\ln n}$ shared memory cells each.
An execution proceeds by repeated attempts to assign names.
During each such an attempt, the processors work to assign tentative names.
Next, the number of distinct tentative names is obtained and if the count equals $n$ then the tentative names become final, otherwise another attempt is made.
We assume that each such an attempt uses a new segment of memory cells \texttt{Counter}  initialized to~$0$s, which is to simplify the exposition and analysis.
An attempt to assign tentative names proceeds by each processor $v$ selecting two integers $\texttt{bin}_v$ and $\texttt{label}_v$ uniformly at random, where $\texttt{bin}_v\in [1,\frac{n}{\ln n}]$ and $\texttt{label}_v\in [1,n^\beta]$.

Next the processors repeatedly attempt to write $\texttt{label}_v$ into \texttt{Bin}$[\texttt{bin}_v]$.
Each such a write is followed by a read and the lucky writer uses the value of memory register \texttt{Counter}$[\texttt{bin}_v]$ to create a pair of numbers $(\texttt{bin}_v,\texttt{Counter}[\texttt{bin}_v])$, after first incrementing \texttt{Counter}$[\texttt{bin}_v]$, which is called $\texttt{bin}_v$'s \emph{position} and is stored in variable $\texttt{position}_v$.
After all processors have their positions determined, we define their ranks as follows.
To find the \emph{rank} of \texttt{position}$_v$, we arrange all such pairs in lexicographic order, comparing first on~\texttt{bin} and then on~$\texttt{Counter}[\texttt{bin}]$, and the rank is the position of this entry in the resulting list, where the first entry has position~$1$, the second~$2$, and so on.

Ranks can be computed using a prefix-type algorithm operating in time~$\cO(\log n)$.
This algorithm first finds for each $\texttt{bin}\in [1,\frac{n}{\ln n}]$ the sum $s(\texttt{bin})=\sum_{1\le i<\texttt{bin}} \texttt{Counter}[i]$.
Next, each processor~$v$ with a position $(\texttt{bin}_v,c)$ assigns to itself $s(\texttt{bin}_v)+c$ as its rank.
After ranks have been computed, they are used as tentative names.

In the analysis of algorithm \textsc{Arbitrary-Unbounded-LV} we will refer to the following bound on independent Bernoulli trials.
Let~$S_n$ be the number of successes in $n$ independent Bernoulli trials, with $p$ as the probability of success.
Let $b(i;n,p)$ be the probability of an occurrence of exactly $i$ successes.
For $r>np$, the following bound holds
\begin{equation}
\label{eqn:S-n-r}
\Pr(S_n\ge r) \le b(r;n,p)\cdot \frac{r(1-p)}{r-np}
\ ,
\end{equation}
see Feller~\cite{Feller-vol1}.

\Paragraph{Balls into bins.}

We consider a process of throwing $n$ balls into $\frac{n}{\ln n}$ bins.
Each ball has a label assigned randomly from the range $[1,n^\beta]$, for $\beta>0$.
We say that a \emph{labeled collision} occurs when there are two balls with the same labels in the same bin.
We refer to this process as the \emph{$\beta$-process}.


\begin{lemma}
\label{lem:balls-arbitrary-unbounded}

For each $a>0$ there exists $\beta>0$ and $c>0$ such that when $n$ balls are labeled with random integers in $[1,n^\beta]$ and next are thrown into $\frac{n}{\ln n}$  bins during the $\beta$-process then there are at most $c\ln n$ balls in every bin and no labeled collision occurs with probability $1-n^{-a}$.
\end{lemma}

\begin{proof}
We estimate from above the probabilities of the events that there are more than $c\ln n$ balls in some bin and that there is a labeled collision.
We show that each of them can be made to be at most $n^{-a}/2$, from which it follows that some of these two events occurs with probability at most~$n^{-a}$.

Let $p$ denote the probability of selecting a specific bin when throwing a ball, which is $p=\frac{\ln n}{n}$.
When we set $r=c\ln n$, for a sufficiently large $c>1$, then 
\begin{equation}
\label{eqn:b-r-n-p-first}
b(r;n,p) = \binom{n}{c\ln n} \Bigl(\frac{\ln n}{n}\Bigr)^{c\ln n} \Bigl(1-\frac{\ln n}{n}\Bigr)^{n-c\ln n}
\ .
\end{equation}
Formula~\eqref{eqn:b-r-n-p-first} translates~\eqref{eqn:S-n-r} into the following bound
\begin{equation}
\label{eqn:S-n-r-translated}
\Pr(S_n\ge r) 
\le 
\binom{n}{c\ln n} \Bigl(\frac{\ln n}{n}\Bigr)^{c\ln n} \Bigl(1-\frac{\ln n}{n}\Bigr)^{n-c\ln n} \cdot \frac{c\ln n(1-\frac{\ln n}{n})}{c\ln n-\ln n}
\ .
\end{equation}
The right-hand side of~\eqref{eqn:S-n-r-translated} can be estimated by the following upper bound:
\begin{align*}
&\ \ \ \ \ \Bigl(\frac{en}{c\ln n}\Bigr)^{c\ln n} \Bigl(\frac{\ln n}{n}\Bigr)^{c\ln n} \Bigl(1-\frac{\ln n}{n}\Bigr)^{n-c\ln n} \cdot \frac{c}{c-1}\\
& = \ 
\Bigl(\frac{e}{c}\Bigr)^{c\ln n} \Bigl(1-\frac{\ln n}{n}\Bigr)^{n} \Bigl(\frac{n}{n-\ln n}\Bigr)^{c\ln n}\cdot \frac{c}{c-1}\\
&\le \ 
n^c c^{-c\ln n} e^{-\ln n}\Bigl(\frac{n}{n-\ln n}\Bigr)^{c\ln n}\cdot \frac{c}{c-1}\\
&\le \ 
n^{-c\ln c +c- 1}
\ ,
\end{align*}
for each sufficiently large $n>0$.
This is because 
\[
\Bigl(\frac{n}{n-\ln n}\Bigr)^{c\ln n} = \Bigl(1+\frac{\ln n}{n-\ln n}\Bigr)^{c\ln n}
\le \exp\Bigl(\frac{c\ln^2 n}{n-\ln n}\Bigr)
\ ,
\] 
which converges to~$1$.
The probability that the number of balls in some bin is greater than $c\ln n$ is therefore at most $n\cdot n^{-c\ln c +c- 1}=n^{-c(\ln c -1)}$, by the union bound.
This probability can be made smaller than $n^{-a}/2$ for a sufficiently large~$c>e$.

The probability of a labeled collision is at most that of a collision when $n$ balls are thrown into~$n^\beta$ bins.
This probability is at most $n^{-\beta+2}$ by bound~\eqref{eqn:probability-collision} used in the proof of Lemma~\ref{lem:balls-arbitrary-bounded}.
This number can be made at most $n^{-a}/2$ for a sufficiently large~$\beta$.
\end{proof}

Next we summarize the performance of algorithm \textsc{Arbitrary-Unbounded-LV} as a Las Vegas algorithm.


\begin{theorem}
\label{thm:arbitrary-unbounded-LV}

Algorithm \textsc{Arbitrary-Unbounded-LV} terminates almost surely and there is no error when the algorithm terminates.
For any $a>0$, there exists $\beta>0$ and  $c>0$ such that the algorithm assigns names within $c\ln n$ time and generates at most $cn\ln n$ random bits with probability at least $1-n^{-a}$.
\end{theorem}

\begin{proof}
The algorithm terminates only when $n$ different names have been assigned, which is provided by the condition that controls the main repeat-loop in Figure~\ref{alg:arbitrary-unbounded-LV}.
This means that there is no error when the algorithm terminates.

We map executions of the $\beta$-process on executions of algorithm \textsc{Arbitrary-Unbounded-LV} in a natural manner.
The main repeat-loop ends after an iteration in which each group of processors that select the same value for the variable~\texttt{bin}, next select distinct values for the variable ~\texttt{label}.
We interpret the random selections in an execution as throwing $n$ balls into $\frac{n}{\ln n}$ bins, where a number \texttt{bin} determines a bin.
The number of iterations of the inner repeat-loop equals the maximum number of balls in a bin.

For any $a>0$, it follows that one iteration of the main repeat-loop suffices with probability at least $1-n^{-a}$, for a suitable~$\beta>0$, by Lemma~\ref{lem:balls-arbitrary-unbounded}.
It follows that $i$ iterations are executed by termination with probability at most $n^{-ia}$, so the algorithm terminates almost surely.

Let us take $c>0$ as in Lemma~\ref{lem:balls-arbitrary-unbounded}.
It follows that an iteration of the main repeat-loop takes at most $c\ln n$ steps and one processor uses at most $c\ln n$ random bits in this one iteration with probability at least $1-n^{-a}$.
\end{proof}

Algorithm \textsc{Arbitrary-Unbounded-LV} is optimal among Las Vegas naming algorithms with respect to the following performance measures: the expected time $\cO(\log n)$, by Theorem~\ref{thm:log-n-lower-bound}, the number of shared memory cells $\cO(n/\log n)$ used to achieve this running time, by Corollary~\ref{cor:four-optima}, both in Section~\ref{sec:lower-bounds-on-time}, and the expected number of used random bits $\cO(n\log n)$, by Proposition~\ref{pro:lower-bound-on-random-bits} in Section~\ref{sec:technical-preliminaries}.

\section{Las Vegas for Common with Bounded Memory }

\label{sec:common-bounded-LV}

We consider the case of Common PRAM when the number of processors $n$ is known and the number of available shared memory cells is constant.
We propose a Las Vegas algorithm called \textsc{Common-Bounded-LV}, whose pseudocode is given in Figure~\ref{alg:common-bounded-LV}.


\begin{figure}[t]
\rule{\textwidth}{0.75pt}

\F 
\textbf{Algorithm} \textsc{Common-Bounded-LV}

\rule{\textwidth}{0.75pt}
\begin{center}
\begin{minipage}{\pagewidth}
\begin{description}
\item[\tt repeat] \ 
\B
\begin{description}
\item[\rm initialize] 
$\texttt{number-of-bins}\gets n$ ; 
$\texttt{name}_v\gets\texttt{Last-Name} \gets 0$  ;\\
$\texttt{no-collision}_v \gets  \texttt{true}$
\item[\tt repeat] \ 
\begin{description}
\item[\rm initialize] $\texttt{Collision-Detected}\gets \texttt{false}$
\item[\tt if] $\texttt{name}_v = 0$ \texttt{then}  \ 
\begin{description}
\item[\texttt{bin}$_v$] $\gets$ random integer  in $[1, \texttt{number-of-bins}]$ 
\item[\tt for]  $i \gets 1$ \texttt{to} $\texttt{number-of-bins} $ \texttt{do}  
\begin{description}
\item[\tt for] $j\gets 1$ \texttt{to} $\beta \ln n$ \texttt{do} 

\texttt{if} $\texttt{bin}_v = i$  \texttt{then} 
\begin{quote}
\texttt{if} \textsc{Verify-Collision} \texttt{then} 

~~~~$\texttt{Collision-Detected}\gets \texttt{collision}_v \gets \texttt{true}$
\end{quote}
\end{description}
\begin{description}
\item[\tt if] $\texttt{bin}_v = i$  \texttt{and not} $\texttt{collision}_v$  \texttt{then} 
\begin{description}
\item[\tt Last-Name] $\gets \texttt{Last-Name} + 1$ 
\item[\texttt{name}$_v$]$ \gets$ \texttt{Last-Name} 
\end{description}
\end{description}
\end{description}
\item[\tt if] $n -\texttt{Last-Name} > \beta \ln n$ 
\begin{description}
\item[\texttt{then}] $\texttt{number-of-bins}\gets (n -\texttt{Last-Name})$ 
 
\item[\texttt{else}] $\texttt{number-of-bins}\gets n/(\beta \ln n)$
 \end{description}
\end{description}
\item[\tt until]  \texttt{not Collision-Detected}
\end{description}
\item[\tt until] $\texttt{Last-Name}= n$ 
\end{description}
\end{minipage}
\FFF

\rule{\textwidth}{0.75pt}

\parbox{\captionwidth}{\caption{\label{alg:common-bounded-LV}
A pseudocode for a processor~$v$ of a Common PRAM, where there is a constant number of shared memory cells.
Procedure \textsc{Verify-Collision} has its pseudocode in Figure~\ref{proc:verify-collision}; lack of parameter means the default parameter~$1$.
The variables \texttt{Collision-Detected} and \texttt{Last-Name} are shared.
The private variable \texttt{name} stores the acquired name.
The constant $\beta$ is a parameter to be determined by analysis.
}}
\end{center}
\end{figure}

An execution of the algorithm is organized as repeated ``attempts'' to assign temporary names.
During such an attempt, each processor without a name chooses uniformly at random an integer in the interval $[1,\texttt{number-of-bins}]$, where \texttt{number-of-bins} is a parameter initialized to $n$; such a selection is interpreted  in a probabilistic analysis as throwing a ball into \texttt{number-of-bins} many bins.
Next, for each $i\in [1,\texttt{number-of-bins}]$, the processors that selected~$i$, if any, verify if they are unique in their selection of $i$ by executing procedure  \textsc{Verify-Collision}  (given in Figure~\ref{proc:verify-collision} in Section~\ref{sec:technical-preliminaries}) $\beta \ln n$ times, where $\beta>0$ is a number that is determined in analysis.

After no collision has been detected, a processor that selected $i$ assigns itself a consecutive name by reading and incrementing the shared variable \texttt{Last-Name}.
It takes up to $\beta \,\texttt{number-of-bins}\ln n$ verifications for collisions for all integers in $[1,\texttt{number-of-bins}]$.
When this is over, the value of variable \texttt{number-of-bins} is modified by decrementing it by the number of new names just assigned, when working with the last \texttt{number-of-bins}, unless such decrementing would result in a number in \texttt{number-of-bins} that is at most $\beta\ln n$, in which case the variable \texttt{number-of-bins} is set to $n/(\beta \ln n)$.
An attempt ends when all processors have tentative names assigned.

These names become final when there are a total of $n$ of them, otherwise there are duplicates, so another attempt is performed.
The main repeat loop in the pseudocode in Figure~\ref{alg:common-bounded-LV} represents an attempt to assign tentative names to each processor.
An iteration of the inner repeat loop during which $\texttt{number-of-bins}>n/(\beta\ln n)$ is called \emph{shrinking} and otherwise it is called \emph{restored}.

\Paragraph{Balls into bins.}

As a preparation for the analysis of performance of algorithm \textsc{Common-Bounded-LV}, we consider a related process of repeatedly throwing balls into bins, which we call the \emph{$\beta$-process}.
The $\beta$-process proceeds through \emph{stages}, each representing one iteration of the inner repeat-loop in Figure~\ref{alg:common-bounded-LV}.
A stage results in some balls removed and some transitioning to the next stage, so that eventually no balls remain and the process terminates.

The balls that participate in a stage are called \emph{eligible} for the stage.
In the first stage, $n$ balls are eligible and we throw $n$ balls into $n$ bins.
Initially, we apply the principle that after all eligible balls have been placed into bins during a stage, the singleton bins along with the balls in them are removed. 
A stage after which bins are removed is called \emph{shrinking}.
There are $k$ bins and $k$ balls in a shrinking stage; we refer to $k$ as the  \emph{length} of this stage.
Given balls and bins for any stage, we choose a bin uniformly at random and independently for each ball in the beginning of a stage and next place the balls in their selected destinations.
The bins that either are empty or multiple in a shrinking stage stay for the next stage.
The balls from multiple bins become eligible for the next stage.

This continues until such a shrinking stage after which at most $\beta \ln n$ balls remain.
Then we restore bins for a total of $n/(\beta \ln n))$ of them to be used in the following stages, during which we never remove any bin; these stages are called \emph{restored}.
In these final restored stages, we keep removing singleton balls at the end of a stage, while balls from multiple bins stay as eligible for the next restored stage.
This continues until all balls are removed. 


\begin{lemma}
\label{lem:balls-common-bounded-LV}

For any $a>0$, there exists $\beta>0$ such that the sum of lengths of all shrinking stages in the $\beta$-process is at most~$2e n$, where $e$ is the base of natural logarithms, and there are at most $\beta \ln n$ restored stages, both events holding with probability $1-n^{-a}$, for sufficiently large~$n$.
\end{lemma}

\begin{proof}
We consider two cases depending on the kind of analyzed stages.
Let $k\le n$ denote the length of a stage.

In a shrinking stage, we throw $k$ balls into $k$ bins, while choosing bins independently and uniformly at random.
The probability that a ball ends up singleton can be bounded from below as follows:
\begin{eqnarray*}
k \cdot \frac{1}{k} \Bigl(1-\frac{1}{k}\Bigr)^{k-1} 
&\ge& 
(e^{-\frac{1}{k}-\frac{1}{k^2}})^{k-1}\\
&=&
e^{-\frac{k-1}{k}-\frac{k-1}{k^2}}\\
&=&
e^{-1+\frac{1}{k}-\frac{1}{k}+\frac{1}{k^2}}\\
&=&
\frac{1}{e}\cdot e^{1/k^2}\\
&\ge& 
\frac{1}{e}
\ ,
\end{eqnarray*}
where we used the inequality $1-x\ge e^{-x-x^2}$, which holds for $0\le x\le \frac{1}{2}$.

Let $Z_k$ be the number of singleton balls after $k$ balls are thrown into $k$ bins.
It follows that the expectancy of $Z_k$ satisfies $\mE [Z_k] \ge k /e$.

To estimate the deviation of $Z_k$ from its expected value, we use the bounded differences inequality~\cite{McDiarmid89,MitzenmacherUpfal-book05}.
Let $B_j$ be the bin of ball $b_j$, for $1\le j\le k$.
Then $Z_k$ is of the form $Z_k=h(B_1,\ldots,B_{k})$ where $h$ satisfied the Lipschitz condition with constant $2$, because moving one ball to a different bin results in changing the value of $h$ by at most $2$ with respect to the original value.
The bounded-differences inequality specialized to this instance is as follows, for any  $d>0$:
\begin{equation}
\label{eqn-E-Z-d}
\Pr(Z_k \le \mE [Z_k] -d \sqrt{k} )
\le
\exp(-d^2/8)
\ .
\end{equation}
We use this inequality for $d=\frac{\sqrt{k}}{2e}$.
Then \eqref{eqn-E-Z-d} implies the following bound:
\begin{eqnarray*}
\Pr\Bigl(Z_k \le \frac{k}{e} -  \frac{k}{2e}\Bigr)
&=&
\Pr\Bigl(Z_k \le \frac{k}{2e}\Bigr)\\
&\le&
\exp\Bigl(-\frac{1}{8}\cdot\Bigl( \frac{\sqrt{k}}{2e}\Bigr)^2\Bigr)\\
&=&
\exp\Bigl(-\frac{k}{32 e^2}\Bigr)
\ .
\end{eqnarray*}
If we start a shrinking stage with $k$ eligible balls  then the number of balls  eligible for the next stage is at most
\[
\Bigl(1-\frac{1}{2e}\Bigr)\cdot k=\frac{2e-1}{2e}\cdot k
\ ,
\]
with probability at least $1-\exp(-k/32e^2)$.
Let us continue shrinking stages as long as the inequality 
\[
\frac{k}{32e^2}> 3a\ln n
\]
holds.
We denote this inequality concisely as $k>\beta \ln n$ for $\beta=96 e^2 a$.
Then the probability that every shrinking stage results in the size of the pool of eligible balls decreasing by a factor of at least 
\[
\frac{2e-1}{2e}=\frac{1}{f}
\]
is itself at least
\[
\Bigl( 1-e^{-3a\ln n}\Bigr)^{\log_f n}
\ge 1-\frac{\log_f n}{n^{-3a}}
\ge 1-n^{-2a}
\ ,
\]
for sufficiently large $n$, by Bernoulli's inequality.

If all shrinking stages result in the size of the pool of eligible balls decreasing by a factor of at least~$1/f$, then the total number of eligible balls summed over all such stages is at most
\[
n\sum_{i\ge 0}f^{-i} = n \cdot \frac{1}{1-f^{-1}} =2en
\ .
\]

In a restored stage, there are at most $\beta \ln n$ eligible balls.
A restored stage happens to be the last one when all the balls become single after their placement, which occurs with probability at least
\[
\Biggl( \frac{n/(\beta\ln n) - \beta\ln n}{n/(\beta\ln n)} \Biggr)^{\beta\ln n}
=
\Bigl( 1-\frac{\beta^2\ln^2 n}{n}\Bigr)^{\beta\ln n}
\ge
1-\frac{\beta^3\ln^3 n}{n}
\ ,
\]
by the Bernoulli's inequality.
It follows that there are more than $\beta\ln n$ restored stages with probability at most 
\[
\Bigl(\frac{\beta^3\ln^3 n}{n}\Bigr)^{\beta\ln n} =n^{-\Omega(\log n)}
\ .
\]
This bound is at most $n^{-2a}$ for sufficiently large $n$.

Both events, one about shrinking stages and the other about restored stages, hold with probability at least $1-2n^{-2a}\ge 1-n^{-a}$, for sufficiently large $n$.
\end{proof}

Next we summarize the performance of algorithm \textsc{Common-Bounded-LV} as  Las Vegas one.
In its proof, we rely on mapping executions of the $\beta$-process on executions of algorithm \textsc{Common-Bounded-LV} in a natural manner.


\begin{theorem}
\label{thm:common-bounded-LV}

Algorithm \textsc{Common-Bounded-LV} terminates almost surely and there is no error when the algorithm terminates.
For  any $a>0$ there exist $\beta>0$ and $c>0$ such that the algorithm terminates within  time $c n\ln n$ using at most $cn\ln n$ random bits with probability $1-n^{-a}$.
\end{theorem}

\begin{proof}
The condition controlling the main repeat-loop guarantees that an execution terminates only when the assigned names  fill the interval $[1,n]$, so they are distinct and there is no error.

To analyze time performance, we consider the $\beta$-process of throwing balls into bins as considered in Lemma~\ref{lem:balls-common-bounded-LV}.
Let $\beta_1>0$ be the number $\beta$ specified in this Lemma, as determined by~$a$ replaced by~$2a$ in its assumptions.
This Lemma gives that the sum of all values of $K$ summed over all shrinking stages is at most $2en$ with probability at least $1-n^{-2a}$.

For a given $K$ and a number $i\in [1,K]$, procedure \textsc{Verify-Collision} is executed $\beta\ln n$ times, where $\beta$ is the parameter in Figure~\ref{alg:common-bounded-LV}.
If there is a collision then it is detected with probability at least $2^{-\beta\ln n}$.
We may take $\beta_2\ge \beta_1$ sufficiently large so that the inequality $2en\cdot 2^{-\beta_2\ln n}<n^{-2a}$ holds.

The total number of instances of executing \textsc{Verify-Collision} during an iteration of the main loop, while $K$ is kept equal to $n/(\beta \ln n)$, is at most $n$.
Observe that the inequality $n\cdot 2^{-\beta_2\ln n}< n^{-2a}$ holds with probability at most $1-n^{-2a}$ because $n< 2en$.

If $\beta$ is set in Figure~\ref{alg:common-bounded-LV} to $\beta_2$ then one iteration of the outer repeat-loop suffices with probability at least $1-2n^{-2a}$, for sufficiently large $n$.
This is because verifications for collisions detect all existing collisions with  this probability.
Similarly, this one iteration takes $\cO(n\log n)$ time with probability that is at least  $1-2n^{-2a}$, for sufficiently large $n$.
The claimed performance holds therefore with  probability at least $1-n^{-a}$, for sufficiently large $n$.

There are at least $i$ iterations of the main repeat-loop with probability at most $n^{-ia}$, so the algorithm terminates almost surely.
\end{proof}

Algorithm \textsc{Common-Bounded-LV} is optimal among Las Vegas algorithms with respect to the following performance measures: the expected time $\cO(n\log n)$, given the amount $\cO(1)$ of its available shared memory, by Corollary~\ref{cor:four-optima} in Section~\ref{sec:lower-bounds-on-time}, and the expected number of random bits $\cO(n\log n)$, by Proposition~\ref{pro:lower-bound-on-random-bits} in Section~\ref{sec:technical-preliminaries}.

\section{Las Vegas for Common with Unbounded Memory}

\label{sec:common-unbounded-LV}

We consider now the last case when the number of processors $n$ is known.
The PRAM is of its Common variant, and there is an unbounded amount of shared memory.
We propose a Las Vegas algorithm called \textsc{Common-Unbounded-LV}, the pseudocode for this algorithm  is given in Figure~\ref{alg:common-unbounded-LV}.
Subroutines of prefix-type, like computing the number of selects and ranks of selected numbers are not included in this pseudocode.
The algorithm invokes procedure \textsc{Verify-Collision}, whose pseudocode is in Figure~\ref{proc:verify-collision}.

An execution of algorithm \textsc{Common-Unbounded-LV} proceeds as a sequence of attempts to assign temporary names.
When such an attempt results in assigning temporary names without duplicates then these transient names become final.
An attempt begins from each processor selecting an integer from the interval $[1,(\beta + 1)n]$ uniformly at random and independently, where $\beta$ is a parameter such that only $\beta>1$ is assumed.
Next, for $\lg n$ steps, each process executes procedure \textsc{Verify-Collision}$(x)$ where $x$ is the currently selected integer.
If a collision is detected then a processor immediately selects another number in $[1,(\beta + 1)n]$ and continues verifying for a collision.
After $\lg n$ such steps, the processors count the total number of selections of different integers.
If this number equals exactly $n$ then the ranks of the selected integers are assigned as names, otherwise another attempt to find names is made.
Computing the number of selections and the ranks takes time $\cO(\log n)$.
In order to amortize this time $\cO(\log n)$ by verifications, such a computation of ranks is performed only after $\lg n$ verifications.
Here a rank of a selected $x$ is the number of selected numbers that are at most~$x$. 


\begin{figure}[t]
\rule{\textwidth}{0.75pt}

\F 
\textbf{Algorithm} \textsc{Common-Unbounded-LV}

\rule{\textwidth}{0.75pt}
\begin{center}
\begin{minipage}{\pagewidth}
\begin{description}
\item[$\texttt{bin}_v\gets$] random integer in $[1,(\beta+1) n ]$ 
\hfill /$\ast$ \ throw a ball into $\texttt{bin}_v$ \ $\ast$/
\B
\item[\tt repeat] \ 
\B
\begin{description}
\Item[\tt for] $i\gets 1$ \texttt{to} $\lg n$ \texttt{do} 
\begin{description}
\item[\tt if] \textsc{Verify-Collision}\,$(\texttt{bin}_v)$ \texttt{then} 
\begin{description}
\item[$\texttt{bin}_v\gets$] random integer in $[1,(\beta+1) n ]$
\end{description}
\end{description}
\item[\tt number-occupied-bins] $\gets$ total number of currently selected values for $\texttt{bin}_v$
\end{description}
\item[\tt until] $\texttt{number-occupied-bins}=n$ 
\B
\item[\tt name$_v$]\!$\gets$ the rank of $\texttt{bin}_v$ among nonempty bins
\end{description}
\end{minipage}
\FFF

\rule{\textwidth}{0.75pt}

\parbox{\captionwidth}{\caption{\label{alg:common-unbounded-LV}
A pseudocode for a processor~$v$ of a Common PRAM, where the number of shared memory cells is unbounded.
The constant~$\beta$ is a parameter that satisfies the inequality $\beta >1$.
The private variable \texttt{name} stores the acquired name.
}}
\end{center}
\end{figure}

\Paragraph{Balls into bins.}

We consider auxiliary processes of placing balls into bins that abstracts operations on shared memory as performed by algorithm \textsc{Common-Unbounded-LV}.

The \emph{$\beta$-process} is about placing $n$ balls into $(\beta+1)n$ bins.
The process is structured as a sequence of stages.
A stage represents an abstraction of one iteration of the inner for-loop in Figure~\ref{alg:common-unbounded-LV} performed as if collisions were detected instantaneously and with certainty.
When a ball is moved then it is placed in a bin selected uniformly at random, all such selections independent from one another.
The stages are performed as follows.
In the first stage, $n$ balls are placed into $(\beta+1)n$ bins.
When a bin is singleton in the beginning of a stage then the ball in the bin stays put through the stage.
When a bin is multiple in the beginning of a stage, then all the balls in this bin participate actively in this stage: they are removed from the bin and placed in randomly-selected bins.
The process terminates after a stage in which all balls reside in singleton bins.

In analysis, it is convenient to visualize a stage as occurring by first removing all balls from multiple bins and then placing the removed balls in randomly selected bins one by one.
We model placements of single balls by movements of a random walk.
More precisely, we associate a \emph{mimicking walk} to each execution of the $\beta$-process.
Such a walk is performed on points with integer  coordinates on a line, as explained next in detail.

The mimicking walk proceeds through stages, similarly as the ball process.
When we are to relocate $k$ balls in a stage of the ball process then this is represented by the mimicking walk starting the corresponding stage at coordinate~$k$.
Suppose that we process a ball in a stage and the mimicking walk is at some position~$i$.
Placing this ball in an empty bin temporarily decreases the number of balls for the next stage; the respective action in the mimicking walk is to decrement its position from $i$ to~$i-1$.
Placing this ball in an occupied bin temporarily increases the number of balls for the next stage; the respective action in the mimicking walk is to increment its position from $i$ to~$i+1$. 
The mimicking walk gives a conservative estimates on the behavior of the ball process, as we show next.


\begin{lemma}
\label{lem:mimicking-walk}

If a stage of the mimicking walk ends at a position~$k$, then the corresponding stage of the ball  $\beta$-process ends with at most $k$ balls to be relocated into bins in the next stage.
\end{lemma}

\begin{proof}
The argument is broken into three cases, in which we consider what happens in the ball $\beta$-process and what are the corresponding actions in the mimicking walk.
A number of balls in a bin in a stage is meant to be the final number of balls in this bin at the end of the stage.

In the first case, just one ball is placed in a bin that begins the stage as empty. 
Then this ball will not be relocated in the next stage.
This means that the number of balls for the next stage decreases by~$1$.
At the same time, the mimicking walk decrements its position by~$1$.

In the second case, some $j\ge 1$ balls land in a bin that is singleton at the start of this stage, so this ball was not eligible for the stage.
Then the number of balls in the bin becomes $j+1$ and these many balls will need to be relocated in the next stage.
Observe that this contributes to incrementing the number of the eligible balls in the next stage by~$1$, because only the original ball residing in the singleton bin is added to the set of eligible balls, while the other balls participate in both stages.
At the same time, the mimicking walk increments its position $j$~times, by~$1$ each time.

In the third and final case, some $j\ge 2$ balls land in a bin that is empty at the start of this stage.
Then this does not contribute to a change in the number of balls eligible for relocation in the next stage, as these $j$ balls participate in both stages.
Let us consider these balls as placed in the bin one by one.
The first ball makes the mimicking walk decrement its position, as the ball is single in the bin.
The second ball makes the walk increment its position, so that it returns to the original position as at the start of the stage.
The following ball placements, if any, result in the walk incrementing its positions.
\end{proof}

\Paragraph{Random walks.}

Next we consider a random walk which will estimate the behavior of a ball process.
One component of estimation is provided by Lemma~\ref{lem:mimicking-walk}, in that we will interpret a random walk as a mimicking walk for the ball process.

The random walk is represented as movements of a marker placed  on the non-negative side of the integer number line.
The movements of the marker are by distance~$1$ at a time and they are independent from each other.
The \emph{random $\beta$-walk} has the marker's position incremented with probability~$\frac{1}{\beta+1}$ and decremented with probability~$\frac{\beta}{\beta+1}$.
This may be interpreted as a sequence of independent Bernoulli trials, in which $\frac{\beta}{\beta+1}$ is chosen to be the probability of success.
We will consider $\beta>1$, for which $\frac{\beta}{\beta+1}>\frac{1}{\beta+1}$, which means that the probability of success is greater than the probability of failure.

Such a random $\beta$-walk proceeds through \emph{stages}, which are defined as follows.
The first stage begins at position~$n$.
When a stage begins at a position $k$ then it ends after $k$ moves, unless it reaches the zero coordinate in the meantime.
The zero point acts as an absorbing barrier, and when the walk's position reaches it then the random walk terminates.
This is the only way in which the walk terminates.
A stage captures one round of PRAM's computation and the number of moves in a stage represents the number of writes processors perform in a round.


\begin{lemma}
\label{lem:random-c-walk}

For any numbers $a>0$ and $\beta> 1$, there exists $b>0$ such that the random $\beta$-walk starting at position $n>0$ terminates within $b \ln n$ stages with all of them comprising $\cO(n)$ moves with probability at least~$1-n^{-a}$.
\end{lemma}

\begin{proof}
Suppose the random walk starts at position $k>0$ when a stage begins.
Let $X_k$ be the number of moves towards~$0$ and $Y_k=k-X_k$ be the number of moves away from~$0$ in such a stage.
The total distance covered towards~$0$, which we call a \emph{drift}, is 
\[
L(k)=X_k-Y_k= X_k-(k-X_k)= 2X_k - k\ .
\]

The expected value of $X_k$ is $\mE [X_k] =\frac{\beta k}{\beta+1}=\mu_k$.
The event $X_k< (1-\varepsilon)\mu_k$ holds with probability at most $\exp(-\frac{\varepsilon^2}{2}\mu_k)$, by the Chernoff bound~\cite{MitzenmacherUpfal-book05}, so that $X_k\ge (1-\varepsilon)\mu_k$ occurs with probability at least $1-\exp(-\frac{\varepsilon^2}{2}\mu_k)$.
We say that such a stage is \emph{conforming} when the event $X_k\ge (1-\varepsilon)\mu_k$ holds.

If a stage is conforming then the following inequality holds:
\[
L(k)
\ge 
2(1-\varepsilon) \frac{\beta k}{\beta+1} - k
= 
\frac{\beta-2\beta\varepsilon-1}{\beta+1} k 
\ .
\]
We want the inequality $\frac{\beta-2\beta\varepsilon-1}{\beta+1}>0$ to hold, which is the case when $\varepsilon < \frac{\beta-1}{2\beta}$.
Let us fix such $\varepsilon > 0$.
Now the distance from $0$ after $k$ steps starting at $k$ is
\[
k-L(k)= (1-\frac{\beta-2\beta\varepsilon-1}{\beta+1})\cdot k
=
\frac{2(1+\beta\varepsilon)}{\beta+1}\cdot k
\ ,
\]
where $\frac{2(1+\beta\varepsilon)}{\beta+1}<1$ for $\varepsilon < \frac{\beta-1}{2\beta}$.
Let $\rho = \frac{\beta+1}{2(1+\beta\varepsilon)}> 1$.
Consecutive $i$ conforming stages make the distance from $0$ decrease by at least a factor $\rho^{-i}$.

When we start the first stage at position $n$ and the next $\log_\rho n$ stages are conforming then after these many stages the random walk ends up at a position that is close to~$0$.
For our purposes, it suffices that the position is of distance at most $s \ln n$ from~$0$, for some $s>0$, because of its impact on the probability estimates.
Namely, the event that all these stages are conforming and the bound~$s \ln n$ on distance from~$0$ holds, occurs with probability at least
\[
1- \log_\rho n \cdot \exp(-\frac{\varepsilon^2}{2}\frac{\beta}{\beta+1}s \ln n)
\ge
1- \log_\rho n \cdot n^{-\frac{\varepsilon^2}{2}\frac{\beta}{\beta+1}s}
\ .
\]
Let us choose $s>0$ such that 
\[
\log_\rho n \cdot n^{-\frac{\varepsilon^2}{2}\frac{\beta}{\beta+1}s}\le \frac{1}{2n^a}
\ ,
\]
for sufficiently large~$n$.

Having fixed $s$, let us take $t>0$ such that the distance covered towards $0$ is at least $s\ln n$ when starting from $k=t\ln n$ and performing $k$ steps.
We interpret these movements as if this was a single conceptual stage for the sake of the argument, but its duration comprises all stages when we start from $s\ln n$ until we terminate at~$0$.
It follows that the conceptual stage comprises at most $t\ln n$ real stages, because a stage takes at least one round.

If this last conceptual stage is conforming then the distance covered towards $0$ is bounded by
\[
L(k)\ge \frac{\beta-2\beta\varepsilon-1}{\beta+1}\cdot k
\ .
\]
We want this to be at least $s\ln n$ for $k=t\ln n$, which is equivalent to
\[
\frac{\beta-2\beta\varepsilon-1}{\beta+1}\cdot t > s
\ .
\]
Now it is sufficient to take $t> s\cdot \frac{\beta+1}{\beta-2\beta\varepsilon-1}$.
This last conceptual stage is not conforming with probability at most $\exp(-\frac{\varepsilon^2}{2}\frac{\beta}{\beta+1}t\ln n)$.
Let us take $t$ that is additionally big enough for the following inequality
\[
\exp(-\frac{\varepsilon^2}{2}\frac{\beta}{\beta+1}t\ln n)
=
n^{-\frac{\varepsilon^2}{2}\frac{\beta}{\beta+1}t}
\le
\frac{1}{2n^a}
\]
to hold.

Having selected $s$ and $t$, we can conclude that there are at most $(s+t)\ln n$ stages with  probability at least $1-n^{-a}$.

Now, let us consider only the total number of moves to the left $X_m$ and to the right $Y_m$ after $m$ moves in total, when starting at position~$n$.
The event $X_m< (1-\varepsilon)\cdot \frac{\beta}{1+\beta}\cdot m$ holds with probability at most $\exp(-\frac{\varepsilon^2}{2}\frac{\beta}{1+\beta}\cdot m)$, by the Chernoff bound~\cite{MitzenmacherUpfal-book05}, so that $X_m\ge m\cdot \frac{(1-\varepsilon) \beta}{1+\beta}$ occurs with at least the respective high probability $1-\exp(-\frac{\varepsilon^2}{2}\frac{\beta}{1+\beta}\cdot m)$.
At the same time, we have that the number of moves away from zero, which we denote~$Y_m$, can be estimated to be
\[
Y_m=m-X_m<m-m\cdot \frac{(1-\varepsilon) \beta}{1+\beta} = \frac{1+\varepsilon \beta}{1+\beta}\cdot m
\ .
\]
This gives an estimate on the corresponding drift:
\[
L(m)=X_m-Y_m>\frac{\beta-2\beta\varepsilon-1}{\beta+1}\cdot m 
\ .
\]
We want the inequality $\frac{\beta-2\beta\varepsilon-1}{\beta+1}>0$ to hold, which is the case when $\varepsilon < \frac{\beta-1}{2\beta}$.
The drift is at least~$n$, with the corresponding large probability, when $m=d\cdot n$ for $d=\frac{\beta+1}{\beta-2\beta\varepsilon-1}$.
The drift is at least such with probability exponentially close to~$1$ in~$n$, which is at least $1-n^{-a}$ for sufficiently large~$n$.
\end{proof}


\begin{lemma}
\label{lem:when-ball-process-terminates}

For any numbers $a>0$ and $\beta> 1$, there exists $b>0$ such that the $\beta$-process starting at position $n>0$ terminates within $b \ln n$ stages after performing $\cO(n)$ ball throws with probability at least~$1-n^{-a}$.

\end{lemma}

\begin{proof}
We estimate the behavior of the $\beta$-process on $n$ balls by the behavior of the random $\beta$-walk starting at position~$n$.
The justification of the estimation is in two steps. 
One is the property of mimicking walks given as Lemma~\ref{lem:mimicking-walk}.
The other is provided by Lemma~\ref{lem:random-c-walk} and is justified as follows.
The  probability of decrementing and incrementing position in the random $\beta$-walk are such that they reflect the probabilities of landing in an empty bin or in an occupied bin. 
Namely, we use the facts that during executing the $\beta$-process, there are at most $n$ occupied bins and at least $\beta\cdot n$ empty bins in any round.
In the $\beta$-process, the probability of landing in an empty bin is at least $\frac{\beta n}{(\beta+1)n}=\frac{\beta}{\beta+1}$, and the probability of landing in an occupied bin is at most~$\frac{n}{(\beta+1)n}=\frac{1}{\beta+1}$.
This means that the random $\beta$-walk is consistent with Lemma~\ref{lem:mimicking-walk} in providing estimates on the time of termination of the $\beta$-process from above.
\end{proof}

\Paragraph{Incorporating verifications.}

We consider the \emph{random $\beta$-walk with verifications}, which is defined as follows.
The process proceeds through stages, similarly as the regular random $\beta$-walk.
For any round of the walk and a position at which the walk is at, we first perform a Bernoulli trial with the probability~$\frac{1}{2}$ of success.
Such a trial is referred to as a \emph{verification}, which is \emph{positive} when a success occurs,  otherwise it is \emph{negative}.
A positive verification results in a movement of the marker  as in the regular $\beta$-walk, otherwise the walk pauses at the given position for this round.


\begin{lemma}
\label{lem:random-c-walk-with-verifications}

For any numbers $a>0$ and $\beta> 1$, there exists $b>0$ such that the random $\beta$-walk with verifications starting at position $n>0$ terminates within $b \ln n$ stages  with all of them comprising the total of $\cO(n)$ moves with probability at least~$1-n^{-a}$.
\end{lemma}

\begin{proof}
We provide an extension of the proof of Lemma~\ref{lem:random-c-walk}, which states a similar property of regular random $\beta$-walks. 
That proof estimated times of stages and the number of moves.
Suppose the regular random $\beta$-walk starts at a position $k$, so that the stage takes $k$ moves.
There is a constant $d<1$ such that the walk ends at a position at most $dk$ with probability exponential in~$k$.

Moreover, the proof of Lemma~\ref{lem:random-c-walk} is such that all the values of $k$ considered are at least logarithmic in~$n$, which provides at most a polynomial bound on error.
A random walk with verifications is slowed down by negative verifications.
Observe that a random walk with verifications that is performed $3k$ times undergoes at least $k$ positive verifications with probability exponential in~$k$ by the Chernoff bound~\cite{MitzenmacherUpfal-book05}.
This means that the proof of Lemma~\ref{lem:random-c-walk} can be adapted to the case of random walks with verifications almost verbatim, with the modifications contributed by polynomial bounds on error of estimates of the number of positive verifications in stages.
\end{proof}

Next, we consider a \emph{$\beta$-process with verifications}, which is  defined as follows.
The process proceeds through stages, similarly as the regular ball process.
The first stage starts with placing $n$ balls into $(\beta+1)n$ bins.
For any following stage, we first go through multiple bins and, for each ball in such a bin, we perform a Bernoulli trial with  the probability~$\frac{1}{2}$ of success, which we call a \emph{verification}.
A success in a trial is referred to as a \emph{positive verification} otherwise it is a \emph{negative} one.
If at least one positive verification occurs for a ball in a multiple bin then all the balls in this bin are relocated in this stage to bins selected uniformly at random and independently for each such a ball, otherwise the balls stay put in this bin until the next stage.
The $\beta$-process terminates when all the balls are singleton.


\begin{lemma}
\label{lem:ball-process-with-verifications-terminates}

For any numbers $a>0$ and $\beta> 1$, there exists $b>0$ such that the $\beta$-process with verifications  terminates within $b \ln n$ stages with all of them comprising the total of $\cO(n)$ ball throws with probability at least~$1-n^{-a}$.
\end{lemma}

\begin{proof}
The argument proceeds by combining Lemma~\ref{lem:mimicking-walk}  with Lemma~\ref{lem:random-c-walk-with-verifications}, similarly as the proof of Lemma~\ref{lem:when-ball-process-terminates} is proved by combining Lemma~\ref{lem:mimicking-walk}  with Lemma~\ref{lem:random-c-walk}.
The details follow.

For any execution of a ball process with verifications, we consider a ``mimicking random walk,'' also with verifications, defined such that when a ball from a multiple bin is handled then the outcome of a random verification for this ball is mapped on a verification for the corresponding random walk.
Observe that for a $\beta$-process with verifications just one positive verification is sufficient among $j-1$ trials when there are $j>1$ balls in a multiple bin, so a random $\beta$-walk with verifications provides an upper bound on time of termination of the $\beta$-process with verifications.
The  probabilities of decrementing and incrementing position in the random $\beta$-walk with verifications are such that they reflect the probabilities of landing in an empty bin or in an occupied bin, similarly as without verifications. 
All this give a consistency of a $\beta$-walk with verifications with Lemma~\ref{lem:mimicking-walk} in providing estimates on the time of termination of the $\beta$-process from above.
\end{proof}

Next we summarize the performance of algorithm \textsc{Common-\-Unbounded-LV} as  Las Vegas one.
The proof is based on mapping executions of the $\beta$-processes with verifications on executions of algorithm \textsc{Common-\-Unbounded-LV} in a natural manner.


\begin{theorem}
\label{thm:common-unbounded-LV}

Algorithm \textsc{Common-\-Unbounded-LV} terminates almost surely and when the algorithm terminates then there is no error.
For each $a>0$ and any $\beta>1$ in the pseudocode, there exists $c>0$ such that the algorithm assigns proper names within time $c\lg n$ and using at most $c n\lg n$ random bits with probability  at least $1-n^{-a}$.
\end{theorem}

\begin{proof}
The algorithm terminates when there are $n$ different ranks, by the condition controlling the repeat-loop.
As ranks are distinct and each in the interval $[1,n]$, each name is unique, so there is no error.
The repeat-loop is executed $\cO(1)$ times with probability at least $1-n^{-a}$, by Lemma~\ref{lem:ball-process-with-verifications-terminates}.
The repeat-loop is performed $i$ times with probability that is at most $n^{-ia}$, so it converges to $0$ with~$i$ increasing.
It follows that the algorithm terminates almost surely.

An iteration of the repeat-loop  in Figure~\ref{alg:common-unbounded-LV} takes $\cO(\log n)$ steps.
This is because of the following two facts.
First, it consists of $\lg n$ iterations of the for-loop, each taking $\cO(1)$ rounds.
Second, it concludes with verifying the until-condition, which is carried out by counting nonempty bins by a prefix-type computation.
It follows that time until termination is $\cO(\log n)$ with probability  $1-n^{-a}$.

By Lemma~\ref{lem:ball-process-with-verifications-terminates}, the total number of ball throws is $\cO(n)$ with probability  $1-n^{-a}$.
Each placement of a ball requires $\cO(\log n)$ random bits, so the number of used random bits is $\cO(n\log n)$ with the same probability.
\end{proof}

Algorithm \textsc{Common-Unbounded-LV} is optimal among Las Vegas naming algorithms with respect to the following performance measures: the expected time $\cO(\log n)$, by Theorem~\ref{thm:log-n-lower-bound}, the number of shared memory cells $\cO(n)$ used to achieve this running time, by Corollary~\ref{cor:four-optima}, and the expected number of random bits  $\cO(n\log n)$, by Proposition~\ref{pro:lower-bound-on-random-bits}.

\section{Monte Carlo for Arbitrary with Bounded Memory}

\label{sec:arbitrary-bounded-MC}

We develop a Monte Carlo naming  algorithm for an Arbitrary PRAM with a constant number of shared memory cells,  when the number of processors $n$ is unknown.
The algorithm is called \textsc{Arbitrary-Bounded-MC} and its pseudocode  is given in Figure~\ref{alg:arbitrary-bounded-MC}.


\begin{figure}[t]
\rule{\textwidth}{0.75pt}

\F 
\textbf{Algorithm} \textsc{Arbitrary-Bounded-MC} 

\rule{\textwidth}{0.75pt}
\begin{center}
\begin{minipage}{\pagewidth}
\begin{description}
\item[\rm initialize] $k\gets 1$ \hfill /$\ast$ \ initial approximation of $\lg n$ \ $\ast$/
\B
\item[\tt repeat] \ 
\B
\begin{description}
\item[\rm initialize] $\texttt{Last-Name} \gets \texttt{name}_v \gets 0$
\item[$k$] $\gets 2k$
\item[$\texttt{bin}_v \gets$] random integer in $[1,2^k]$
\hfill /$\ast$ \ throw a ball into a bin \ $\ast$/
\item[\tt repeat] \ 
\begin{description}
\item[\tt All-Named]$\gets$ \texttt{true}
\item[\tt if] $\texttt{name}_v = 0$ \texttt{then} 
\begin{description}
\item[\tt Pad] $\gets \texttt{bin}_v$
\item[\tt if] $\texttt{Pad}=\texttt{bin}_v$ \texttt{then}
\begin{description}
\item[\tt Last-Name]$\gets \texttt{Last-Name}+1$  
\item[\tt name$_v$]$\gets \texttt{Last-Name}$
\end{description}
\item[\tt else] \ 
\begin{description}
\item[\texttt{All-Named}] $\gets\texttt{false}$
\end{description}
\end{description}
\end{description}

\item[\tt until All-Named]\  

\end{description}
\item[\tt until] $\texttt{Last-Name} \le 2^{k/\beta}$
\end{description}
\end{minipage}
\FFF

\rule{\textwidth}{0.75pt}

\parbox{\captionwidth}{\caption{\label{alg:arbitrary-bounded-MC}
A pseudocode for a processor $v$ of an Arbitrary PRAM with a constant number of shared memory cells.
The variables \texttt{Last-Name}, \texttt{All-Named} and \texttt{Pad} are shared.
The private variable \texttt{name} stores the acquired name.
The constant $\beta>0$ is a parameter to be determined by analysis.
}}
\end{center}
\end{figure}

The underlying idea is to have all processors repeatedly attempt to obtain tentative names and terminate when the probability of duplicate names is gauged to be sufficiently small. 
To this end, each processor writes an integer selected from a suitable ``selection range'' into a shared memory register and next reads this register to verify whether the write was successful or not.
A successful write results in each such a processor getting a tentative name by reading and incrementing another shared register operating as a counter.
One of the challenges here is to determine a selection range from which random integers are chosen for writing.
A good selection range is large enough with respect to the number of writers, which is unknown, because when the range is too small then multiple processors may select the same integer and so all of them get the same tentative name after this integer gets written successfully.
The algorithm keeps the size of a selection range growing with each failed attempt to assign tentative names.

There is an inherent tradeoff here, since on the one hand, we want to keep the size of used shared memory small, as a measure of efficiency of the algorithm, while, at the same time, the larger the range of memory the smaller the probability of collision of random selections from a selection range and so of the resulting duplicate names.
Additionally, increasing the selection range repeatedly costs time for each such a repetition, while we also want to minimize the running time as the metric of performance.
The algorithm keeps increasing the selection range with a quadratic rate, which turns out to be sufficient to optimize all the performance metrics we measure.
The algorithm terminates when the number of selected integers from the current selection range makes a sufficiently small fraction of the size of the used range.

The  structure of the pseudocode in Figure~\ref{alg:arbitrary-bounded-MC} is determined by the main repeat-loop.
Each iteration of this loop begins with doubling the variable~$k$, which determines the selection  range $[1,2^k]$.
This means that the size of the selection range increases quadratically with consecutive iterations of the main repeat-loop.
A processor begins an iteration of the main loop by choosing an integer uniformly at random from the current selection range  $[1,2^k]$.
There is an inner repeat-loop, nested within the main loop, which assigns tentative names depending on the random selections just made.

All processors repeatedly write to a shared variable \texttt{Pad} and next read to verify if the write was successful. 
It is possible that different processors attempt to write the same value and then verify that their write was successful.
The shared variable \texttt{Last-Name} is used to proceed through consecutive integers to provide tentative names to be assigned to the latest successful writers.  
When multiple processors attempt to write the same value to \texttt{Pad} and it gets written successfully, then all of them obtain the same tentative name.
The variable \texttt{Last-Name}, at the end of each iteration of the inner repeat-loop, equals the number of occupied bins.
The shared variable \texttt{All-Named} is  used to verify if all processors have tentative names.
The outer loop terminates when the number of assigned names, which is the same as the number of occupied bins, is smaller than or equal to~$2^{k/\beta}$, where $\beta>0$ is a parameter to be determined in analysis.

\Paragraph{Balls into bins.}

We consider the following auxiliary \emph{$\beta$-process} of throwing balls into bins, for a parameter $\beta>0$.
The process proceeds through stages identified by consecutive positive integers.
The $i$th stage has the number parameter~$k$ equal to $k=2^i$ .
During a stage, we first throw $n$ balls into the corresponding $2^k$ bins and next count the number of occupied bins.
A stage is last in an execution of the $\beta$-process, and so the $\beta$-process terminates, when the number of occupied bins  is smaller  than or equal to~$2^{k/\beta}$.

We may observe that the $\beta$-process always terminates.
This is because, by its specification, the $\beta$-process terminates by the first stage in which the inequality $n\le 2^{k/\beta}$ holds, where $n$ is an upper bound on the number of occupied bins in a stage. 
The inequality $n\le 2^{k/\beta}$ is equivalent to $n^\beta\le 2^{k}$ and so to $\beta  \lg n \le k$.
Since $k$ goes through consecutive powers of~$2$, we obtain that the number of stages of the $\beta$-process with $n$ balls is at most $\lg(\beta \lg n)=\lg\beta + \lg\lg n$.

We say that such a $\beta$-process is \emph{correct} when upon termination each ball is in a separate bin, otherwise the process is \emph{incorrect}.


\begin{lemma}
\label{lem:balls-k-over-beta-first}

For any $a>0$ there exists $\beta>0$ such that the $\beta$-process is incorrect with probability that is at most $n^{-a}$, for sufficiently large~$n$.
\end{lemma}

\begin{proof}
The $\beta$-process is incorrect when  there are collisions after the last stage.  
The probability of the intersection of the events ``$\beta$-process terminates'' and ``there are collisions'' is bounded from above by the probability of each of these two events.
Next we show that, for each pair of $k$ and $n$, some of these two events occurs with probability that is at most $n^{-a}$, for a suitable~$\beta$. 

First, we consider the event that the $\beta$-process terminates.
The probability that there are at most $2^{k/\beta}$ occupied bins is at most
\begin{eqnarray}
\binom{2^k}{2^{k/\beta}}\Bigl( \frac{2^{k/\beta}}{2^k} \Bigr)^n
\nonumber
&\le&
\Bigl( \frac{e 2^k}{2^{k/\beta}} \Bigr)^{2^{k/\beta}} 2^{k(\beta^{-1}-1)n}\\
&\le&
\nonumber
e^{2^{k/\beta}} \cdot 2^{k(1-\beta^{-1})2^{k/\beta}} \cdot 2^{k(\beta^{-1}-1)n}\\
&\le&
\label{eqn:before-logarithm}
e^{2^{k/\beta}} \cdot 2^{k(\beta^{-1}-1)(n-2^{k/\beta})}
\ .
\end{eqnarray}
We estimate from above the natural logarithm of the right-hand side of~\eqref{eqn:before-logarithm}.
We obtain the following upper bound:
\begin{eqnarray}
2^{k/\beta} + k(\beta^{-1}-1)(n-2^{k/\beta})\ln 2
\nonumber
&<&
2^{k/\beta} -\frac{1}{2}(n-2^{k/\beta})\ln 2\\
\nonumber
&=&
2^{k/\beta} -\frac{\ln 2}{2}n + \frac{\ln 2}{2} 2^{k/\beta}\\
&=&
\label{eqn:bound-on-logarithm}
-\frac{\ln 2}{2}n + 2^{k/\beta} \cdot \frac{2+\ln 2}{2} 
\ ,
\end{eqnarray}
for $\beta >4/3$, as $k\ge 2$.
The estimate~\eqref{eqn:bound-on-logarithm} is at most $-n\cdot \frac{\ln 2}{4}$ when $2^{k/\beta}\le n\cdot \delta$, for $\delta=\frac{\ln 2}{2(2+\ln 2)}$, by a direct algebraic verification.
These restrictions on $k$ and $\beta$ can be restated as 
\begin{equation}
\label{eqn:first-condition-on-k}
k\le \beta\lg(n\delta) \text{ and }\beta >4/3
\ .
\end{equation}
When this condition~\eqref{eqn:first-condition-on-k} is satisfied, then the probability of at most $2^{k/\beta}$ occupied bins is at most
\[
\exp\Bigl(-n\cdot \frac{\ln 2}{4}\Bigr)\le n^{-a}
\]
for sufficiently large~$n$.

Next, let us consider the probability of collisions occurring.
Collisions do not occur with probability that is at least
\[
\Bigl(1-\frac{n}{2^k}\Bigr)^n \ge 1-\frac{n^2}{2^k}
\ ,
\]
by the Bernoulli's inequality.
It follows that the probability of collisions occurring can be bounded from above by $\frac{n^2}{2^k}$.
This bound in turn is at most $n^{-a}$ when 
\begin{equation}
\label{eqn:second-condition-on-k}
k\ge (2+a)\lg n
\ .
\end{equation}

In order to have some of the inequalities~\eqref{eqn:first-condition-on-k} and~\eqref{eqn:second-condition-on-k} hold for any $k$ and $n$, it is sufficient to have
\[
(2+a)\lg n \le \beta\lg(n\delta)
\ .
\]
This determines $\beta$ as follows: 
\[
\beta\ge \frac{(2+a)\lg n }{\lg  n+\lg \delta}\rightarrow 2+a
\ ,
\]
with ${n\rightarrow\infty}$.
We obtain that the inequality $\beta > 2+a$ suffices, for $n$ that is large enough.
\end{proof}


\begin{lemma}
\label{lem:balls-k-over-beta-second}

For each $\beta>0$ there exists $c>0$ such that when the $\beta$-process terminates then the number of bins ever needed is at most $cn$ and the number of random bits ever generated is at most $c n \ln n$.
\end{lemma}

\begin{proof}
The $\beta$-process terminates by the stage in which the inequality $n\le 2^{k/\beta}$ holds, so $k$ gets to be at most $\beta\lg n$.
We partition the range $[2, \beta\lg n]$ of values of~$k$ into two subranges and consider them separately.

First, when $k$ ranges from $2$ to $\lg n$ through the stages, then the numbers of needed bins increase quadratically through the stages, because $k$ is doubled with each transition to the next stage.
This means that the total number of all these bins is $\cO(n)$.
At the same time, the number of random bits increases geometrically through the stages, so the total number of random bits a processor uses is $\cO(\log n)$.

Second, when $k$ ranges from $\lg n$ to $\beta\lg n$, the number of needed bins is at most $n$ in each stage.
There are only $\lg(\beta + 1)$ such stages, so the total number of all these bins is $\lg(\beta + 1)\cdot n$.
At the same time, a processor uses at most $\beta  \lg n$ random bits in each of these stages.
\end{proof}

There is a direct correspondence between iterations of the outer repeat-loop and stages of a $\beta$-process.
The $i$th stage has the number~$k$ equal to the value of $k$ during the $i$th iteration of the outer repeat-loop of algorithm \textsc{Arbitrary-Bounded-MC}, that is, we have $k=2^i$.
We map an execution of the algorithm into a corresponding execution of a $\beta$-process in order to apply Lemmas~\ref{lem:balls-k-over-beta-first} and~\ref{lem:balls-k-over-beta-second} in the proof of the Theorem~\ref{thm:arbitrary-bounded-MC}, which  summarizes the performance of algorithm \textsc{Arbitrary-Bounded-MC} and justifies that it is Monte Carlo.


\begin{theorem}
\label{thm:arbitrary-bounded-MC}

Algorithm \textsc{Arbitrary-Bounded-MC} always terminates, for any $\beta>0$.
For each $a>0$ there exists $\beta>0$ and $c>0$ such that  the algorithm assigns unique names, works in time at most~$cn$, and  uses at most $cn\ln n$ random bits, all this with probability at least $1-n^{-a}$.
\end{theorem}

\begin{proof}
The number of stages of the $\beta$-process with $n$ balls is at most $\lg(\beta \lg n)=\lg\beta + \lg\lg n$.
This is also an upper bound on the number of iterations of the main repeat-loop.
We conclude that the algorithm always terminates.

The number of bins available in a stage is an upper bound on the number of bins occupied in this stage.
The number of bins occupied in a stage equals the number of times the inner repeat-loop is iterated, because executing instruction $\texttt{Pad}\gets \texttt{bin}$ eliminates one occupied bin.
It follows that the number of bins ever needed is an upper bound on time of the algorithm.
The number of iterations of the inner repeat-loop is recorded in the variable \texttt{Last-Name}, so the termination condition of the algorithm corresponds to the termination condition of the $\beta$-process.

When the $\beta$-process is correct then this means that the processors obtain distinct names. 
We conclude that  Lemmas~\ref{lem:balls-k-over-beta-first} and~\ref{lem:balls-k-over-beta-second} apply when understood about the behavior of the algorithm.
This implies the following: the names are correct and execution terminates in $\cO(n)$ time while $\cO(n\log n)$ bits are used, all this with probability that is at least $1-n^{-a}$.
\end{proof}

Algorithm \textsc{Arbitrary-Bounded-MC} is optimal with respect to the following performance measures: the expected time $\cO(n)$, by Theorem~\ref{thm:lower-bound-memory-arbitrary-pram}, the expected number of random bits  $\cO(n\log n)$, by Proposition~\ref{pro:lower-bound-on-random-bits},  and the probability of error~$n^{-\cO(1)}$, by Proposition~\ref{pro:probability-of-error}.

\section{Monte Carlo for Arbitrary with Unbounded Memory}

\label{sec:arbitrary-unbounded-MC}

We develop a Monte Carlo naming algorithm for Arbitrary PRAM with an unbounded amount of shared registers, when the number of processors $n$ is unknown.
The algorithm is called \textsc{Arbitrary-Unbounded-MC} and its pseudocode  is given in Figure~\ref{alg:arbitrary-unbounded-MC}.

The underlying idea is to parallelize the process of selection of names applied in Section~\ref{sec:arbitrary-bounded-MC} in algorithm \textsc{Arbitrary-Bounded-MC} so that multiple processes could acquire information in the same round that later would allow them to obtain names.
As algorithm \textsc{Arbitrary-Bounded-MC} used shared registers \texttt{Pad} and \texttt{Last-Name}, the new algorithm uses arrays of shared registers playing similar roles.
The values read-off from \texttt{Last-Name} cannot be uses directly as names, because multiple processors can read the same values, so we need to distinguish between these values to assign names. 
To this end, we assign ranks to processors based on their lexicographic ordering by pairs of numbers determined by \texttt{Pad} and \texttt{Last-Name}.


\begin{figure}[t]
\rule{\textwidth}{0.75pt}

\F 
\textbf{Algorithm} \textsc{Arbitrary-Unbounded-MC} 

\rule{\textwidth}{0.75pt}
\begin{center}
\begin{minipage}{\pagewidth}
\begin{description}
\item[\rm initialize] $k\gets 1$ 
\hfill /$\ast$ \ initial approximation of $\lg n$ \ $\ast$/

\B
\item[\tt repeat] \ 
\B
\begin{description}
\item[\rm initialize] $\texttt{All-Named} \gets \texttt{true}$ 

\item[\rm initialize]  $\texttt{position}_v \gets (0,0)$
\item[$k$] $\gets r(k)$
\item[$\texttt{bin}_v\gets$]  random integer in $[1,2^k/(\beta k)]$  
\hfill /$\ast$ \ choose a bin for the ball \ $\ast$/
\item[$\texttt{label}_v\gets$]  random integer in $[1,2^{\beta k}]$  
\hfill /$\ast$ \ choose a label for the ball \ $\ast$/
\item[\tt for] $i\gets 1$ \texttt{to} $\beta k $ \texttt{do} 
\begin{description}
\item[\tt if] $\texttt{position}_v= (0,0)$ \texttt{then}
\begin{description}
\item[\texttt{Pad}\!]$[\texttt{bin}_v]\gets \texttt{label}_v$
\item[\tt if] $\texttt{Pad}\,[\texttt{bin}_v] = \texttt{label}_v$ \texttt{then}
\begin{description}
\item[\texttt{Last-Name}\!\!]$[\texttt{bin}_v]\gets \texttt{Last-Name}\,[\texttt{bin}_v]+1$  
\item[\texttt{position}$_v$] $\gets (\texttt{bin}_v,\texttt{Last-Name}\,[\texttt{bin}_v])$
\end{description}
\end{description}
\end{description}
\item[\texttt{if}] $\texttt{position}_v= (0,0)$ \texttt{then} 
\begin{description}
\item[\texttt{All-Named}] $\gets \texttt{false}$
\end{description}

\end{description}
\item[\tt until] \texttt{All-Named}  
\B
\item[\tt name$_v$]\!$\gets$  the rank of  \texttt{position}$_v$
\end{description}
\end{minipage}
\FFF

\rule{\textwidth}{0.75pt}

\parbox{\captionwidth}{\caption{\label{alg:arbitrary-unbounded-MC}
A pseudocode for a processor~$v$ of an Arbitrary PRAM, when the number of shared memory cells is unbounded.
The variables \texttt{Pad} and \texttt{Last-Name} are arrays of shared memory cells, the variable \texttt{All-Named} is shared as well.
The private variable \texttt{name} stores the acquired name.
The constant $\beta>0$ and an increasing function $r(k)$ are parameters.
}}
\end{center}
\end{figure}

The pseudocode in Figure~\ref{alg:arbitrary-unbounded-MC} is structured as a repeat-loop.
In the first iteration, the parameter~$k$ equals~$1$, and in subsequent ones is determined by iterations of the increasing integer-valued function~$r(k)$, which is a parameter.
We consider two instantiations of the algorithm, determined by $r(k)=k+1$ and by $r(k)=2k$.
In one iteration of the main repeat-loop, a processor uses two  variables $\texttt{bin}\in[1,2^k/(\beta k)]$ and $\texttt{label}\in [1,2^{\beta k}]$, which are selected independently and uniformly at random from the respective ranges.

We interpret \texttt{bin} as a bin's number and $\texttt{label}$ as a label for a ball.
Processors write their values~$\texttt{label}$ into the respective bin by instruction $\texttt{Pad}\,[\texttt{bin}] \gets \texttt{label}$ and verify what value got written.
After a successful write, a processor increments \texttt{Last-Name}$[\texttt{bin}]$ and assigns the pair $(\texttt{bin},\texttt{Last-Name}\,[\texttt{bin}])$ as its \emph{position}.
This is repeated $\beta k$ times by way of iterating the inner for-loop.
This loop has a specific upper bound $\beta k$ on the number of iterations because we want to ascertain that there are at most $\beta k$  balls in each bin.
The main repeat-loop terminates when all values attempted to be written actually get written.
Then processors assign themselves names according to the ranks of their positions.
The array \texttt{Last-Name} is assumed to be initialized to $0$'s, and in each iteration of the  repeat-loop we use a fresh region of shared memory to allocate this array.

\Paragraph{Balls into bins.}

We consider a related process of placing labeled balls into bins, which is referred to as \emph{$\beta$-process}. 
Such a process proceeds through stages and is parametrized by a function $r(k)$.
In the first stage, we have $k=1$, and given some value of $k$ in a stage, the next stage has this parameter equal to $r(k)$.
In a stage with a given $k$, we place $n$ balls into~$2^k/(\beta k)$ bins, with labels from $[1,2^{\beta k}]$.
The selections of bins and labels are performed independently and uniformly at random.
A stage terminates the $\beta$-process when there are at most $\beta k$ labels of balls in each bin.

\begin{lemma}
\label{lem:beta-process-terminates}

The $\beta$-process always terminates.
\end{lemma}

\begin{proof}
The $\beta$-process terminates by a stage in which the inequality $n\le \beta k$ holds,  because $n$ is an upper bound on the number of balls in a bin. 
This always occurs when function $r(k)$ is increasing.
\end{proof}

We expect the $\beta$-process to terminate earlier, as Lemma~\ref{lem:arbitrary-unbounded-one} states.

\begin{lemma}
\label{lem:arbitrary-unbounded-one}

For each $a>0$, if $k\le \lg n - 2$ and $\beta\ge 1+a$ then the probability of halting in the  stage  is smaller than~$n^{-a}$, for sufficiently large~$n$.
\end{lemma}

\begin{proof}
We show that when $k$ is suitably small then the probability of at most $\beta k$ different labels in each bin is small.
There are $n$ balls placed into $2^k/(\beta k)$ bins, so there are at least $\frac{\beta k n}{2^k}$ balls in some bin, by the pigeonhole principle.
We consider these balls and their labels.

The probability that all these balls have at most $\beta k$ labels is at most
\begin{eqnarray}
\binom{2^{\beta k}}{\beta k} \Bigl( \frac{\beta k}{2^{\beta k}}\Bigr)^{\frac{\beta kn}{2^k}}
\nonumber
&\le&
\Bigl( \frac{e 2^{\beta k}}{\beta k}\Bigr)^{\beta k} \cdot \frac{(\beta k)^{\frac{\beta k n}{2^k}}}{(2^{\beta k})^{\frac{\beta kn}{2^k}}}\\
\nonumber
&=&
e^{\beta k} 2^{\beta k (\beta k-\frac{\beta kn}{2^k})} (\beta k )^{\frac{\beta kn}{2^k}-\beta k}\\
&=&
\label{eqn:beta-k}
e^{\beta k} \Bigl( \frac{\beta k}{2^{\beta k}} \Bigr)^{\frac{\beta k n}{2^k}-\beta k}
\ .
\end{eqnarray}
We want to show that this is at most $n^{-a}$.
We compare the binary logarithms  of $n^{-a}$ and the right-hand side of~\eqref{eqn:beta-k}, and want the following inequality to hold:
\[
\lg e \cdot \beta k + \Bigl(\frac{\beta k n}{2^k}-\beta k\Bigr) (\lg (\beta k) - \beta k ) \le -a\lg n 
\ ,
\]
which is equivalent to the following inequality, by algebra:
\begin{equation}
\label{eqn:a-beta}
\frac{n}{2^k}\ge \frac{\lg e}{ \beta k-\lg(\beta k)} + 1 + \frac{a\lg n}{\beta k( \beta k-\lg(\beta k))}
\ .
\end{equation}
Observe now that, assuming $\beta\ge a+1$, if $k<\sqrt{\lg n}$ then the right-hand side of~\eqref{eqn:a-beta} is at most $\cO(1)+\lg n$ while the left-hand side is at least $\sqrt{n}$, and when $\sqrt{\lg n}\le k\le \lg n-2$ then the right-hand side of~\eqref{eqn:a-beta} is at most~$3$ while the left-hand side is at least~$4$, for sufficiently large~$n$.
\end{proof}

We say that a \emph{label collision} occurs, in a configuration produced by the process, if some bin contains two balls with the same label.

\begin{lemma}
\label{lem:arbitrary-unbounded-two}

For any $a>0$, if $k>\frac{1}{2}\lg n$ and $\beta>4a+7$  then the probability of a label collision is smaller than $n^{-a}$.
\end{lemma}

\begin{proof}
The number of pairs of a bin number and a label is $2^k\cdot 2^{\beta k}/(\beta k)$.
It follows that the probability of some two balls in the same bin obtaining different labels is at least
\[
\Bigl(1-\frac{n}{2^{k+\beta k}/(\beta k)}\Bigr)^n
\ge
1-\frac{n^2}{2^{k+\beta k}/(\beta k)}
\ ,
\]
by the Bernoulli's inequality.
So the probability that two different balls obtain the same label  is at most $\frac{n^2}{2^{k+\beta k}/(\beta k)}$.
We want the following inequality to hold
\[
\frac{n^2}{2^{k+\beta k}/(\beta k)}
<
n^{-a}
\ .
\]
This is equivalent to the inequality obtained by taking logarithms
\[
(2+a)\lg n < (1+\beta)k -\lg(\beta k)
\ ,
\]
which holds when $(2+a)\lg n< \frac{1+\beta}{2}k$.
It follows that it is sufficient for $k$ to satisfy 
\[
k> \frac{2(2+a)}{1+\beta}\lg n
\ .
\]
This inequality holds for $k>\frac{1}{2}\lg n$ when $\beta>4a+7$.
\end{proof}

We say that such a $\beta$-process is \emph{correct} when upon termination no label collision occurs, otherwise the process is \emph{incorrect}.


\begin{lemma}
\label{lem:correctness-arbitrary-unbounded}

For any $a>0$, there exists $\beta>0$ such that the $\beta$-process is incorrect with probability that is at most $n^{-a}$, for sufficiently large~$n$.
\end{lemma}

\begin{proof}
The $\beta$-process is incorrect when  there is a label collision after the last stage.  
The probability of the intersection of the events ``$\beta$-process terminates'' and ``there are label collisions'' is bounded from above by the probability of any one of these events.
Next we show that, for each pair of $k$ and $n$, some of these two events occurs with probability that is at most $n^{-a}$, for a suitable~$\beta$. 

To this end we use Lemmas~\ref{lem:arbitrary-unbounded-one} and~\ref{lem:arbitrary-unbounded-two} in which we substitute $2a$ for $a$.
We obtain that, on the one hand, if $k\le \lg n - 2$ and $\beta\ge 1+2a$ then the probability of halting is smaller than $n^{-2a}$, and, on the other hand, that if $k>\frac{1}{2}\lg n$ and $\beta>8a+7$  then the probability of a label collision is smaller than $n^{-2a}$.
It follows that some of the two considered events occurs with probability at most $2n^{-2a}$ for sufficiently large $\beta$ and any sufficiently large $n$.
This probability is at most $n^{-a}$, for sufficiently large~$n$.
\end{proof}


\begin{lemma}
\label{lem:arbitrary-unbounded-three}

For any $a>0$, there exists $\beta>0$ and $c>0$ such that the following two facts about the $\beta$-process hold.
If $r(k)=k+1$ then at most $cn/\ln n$ bins are ever needed and $c n\ln^2 n$ random bits are ever generated, each among these properties occurring with probability that is at least $1-n^{-a}$.
If $r(k)=2k$ then at most $cn^2/\ln n$ bins are ever needed and $c n\ln n$ random bits are ever generated, each among these properties occurring with probability that is at least $1-n^{-a}$.
\end{lemma}

\begin{proof}
We throw $n$ balls into $2^k/(\beta k)$ bins.
As $k$ keeps increasing, the probability of termination increases as well, because both $2^k/(\beta k)$ and $\beta k$ increase as functions of~$k$.
Let us take $k=1+\lg n$ so that the number of bins is~$\frac{2n}{\beta k}$.
We want to show that no bin contains more than $\beta k$ balls with a suitably small probability.

Let us consider a specific bin and let $X$ be the number of balls in this bin.
The expected number of balls in the bin is $\mu=\frac{\beta k}{2}$.
We use the Chernoff bound for a sequence of Bernoulli trials in the form of 
\[
\Pr(X>(1+\varepsilon)\mu )
<
\exp(-\varepsilon^2\mu/3)
\ ,
\]
which holds for $0<\varepsilon<1$, see~\cite{MitzenmacherUpfal-book05}.
Let us choose $\varepsilon=\frac{1}{2}$, so that $1+\varepsilon=\frac{3}{2}$ and $\frac{3}{2}\mu=\frac{3}{4}\beta k$.
We obtain the following bound
\[
\Pr(X>\beta k ) 
<
\Pr\bigl(X>\frac{3}{4}\cdot \beta k \bigr) 
< 
\exp(-\frac{1}{4}\cdot \frac{\beta k}{6})
=
\exp\bigl(-\frac{\beta}{24}\cdot (1+\lg n)\bigr)
\ ,
\]
which can be made smaller than $n^{-1-a}$ for a $\beta$ sufficiently large with respect to~$a$, and sufficiently large~$n$.
Using the union bound, each of the $n$ bins contains at most $\beta k$ balls with probability at most~$n^{-a}$.
This implies that termination occurs as soon as $k$ reaches or surpasses $k=1+\lg n$, with the corresponding large probability $1-n^{-a}$.

In the case of $r(k)=k+1$, the consecutive integer values of $k$ are tried, so the $\beta$-process terminates by the time $k=1+\lg n$, and for this $k$ the number of bins needed is $\Theta(n/\log n)$.
To choose a bin for any value of $k$ requires at most $k$ random bits, so implementing such choices for $k=1,2,\ldots, 1+\lg n$ requires $\cO(\log^2 n)$ random bits per processor.

In the case of $r(k)=2k$, the $\beta$-process terminates by the time the magnitude of~$k$ reaches $2(1+\lg n)$, and for this value of~$k$ the number of bins needed is $\Theta(n^2/\log n)$.
As $k$ progresses through consecutive powers of $2$, the sum of these numbers is a sum of a geometric progression, and so is of the order of the maximum term, that is $\Theta(\log n)$, which is the number of random bits per processor.
\end{proof}

There is a direct correspondence between iterations of the outer repeat-loop of algorithm \textsc{Arbitrary-Unbounded-MC} and stages of the $\beta$-process.
We map an execution of the algorithm into a corresponding execution of a $\beta$-process in order to apply Lemmas~\ref{lem:correctness-arbitrary-unbounded} and~\ref{lem:arbitrary-unbounded-three} in the proof of the following Theorem, which  summarizes the performance of algorithm \textsc{Arbitrary-Unbounded-MC} and justifies that it is Monte Carlo.


\begin{theorem}
\label{thm:arbitrary-unbounded-MC}

Algorithm \textsc{Arbitrary-Unbounded-MC} always terminates, for any $\beta>0$.
For each $a>0$, there exists $\beta>0$ and $c>0$ such that  the algorithm assigns unique names and has the following additional properties with probability $1-n^{-a}$.
If $r(k)=k+1$ then at most $cn/\ln n$ memory cells are ever needed, $c n\ln^2 n$ random bits are ever generated, and the algorithm terminates in time $\cO(\log^2 n)$.
If $r(k)=2k$ then at most $cn^2/\ln n$ memory cells are ever needed, $c n\ln n$ random bits are ever generated, and the algorithm terminates in time $\cO(\log n)$.
\end{theorem}

\begin{proof}
The algorithm always terminates by Lemma~\ref{lem:beta-process-terminates}.
By Lemma~\ref{lem:correctness-arbitrary-unbounded}, the algorithm assigns correct names with probability that is at least $1-n^{-a}$.
The remaining properties follow from Lemma~\ref{lem:arbitrary-unbounded-three}, because the number of bins is proportional to the number of memory cells and the number of random bits per processor is proportional to time.
\end{proof}

The instantiations of algorithm \textsc{Arbitrary-Unbounded-MC} are close to optimality with respect to some of the performance metrics we consider, depending on whether $r(k)=k+1$ or $r(k)=2k$.
If $r(k)=k+1$ then the algorithm's use of shared memory would be optimal if its time were $\cO(\log n)$, by Theorem~\ref{thm:lower-bound-memory-arbitrary-pram}, but as it is, the algorithm misses space optimality by at most a logarithmic factor, since the algorithm's running time is $\cO(\log^2 n)$.
Similarly, if $r(k)=k+1$ then the number of random bits ever generated $\cO(n\log^2 n)$ misses optimality by at most a logarithmic factor, by Proposition~\ref{pro:lower-bound-on-random-bits}.
On the pother hand, if $r(k)=2k$ then the expected time $\cO(\log n)$ is optimal, by Theorem~\ref{thm:log-n-lower-bound}, the expected number of random bits  $\cO(n\log n)$ is optimal, by Proposition~\ref{pro:lower-bound-on-random-bits},  and the probability of error~$n^{-\cO(1)}$ is optimal, by Proposition~\ref{pro:probability-of-error}, but the amount of used shared memory misses optimality by at most a polynomial factor, by Theorem~\ref{thm:lower-bound-memory-arbitrary-pram}.

\section{Monte Carlo for Common with Bounded Memory}

\label{sec:common-bounded-MC}

The Monte Carlo algorithm \textsc{Common-Bounded-MC}, which we present in this section, solves the naming problem for a Common PRAM with a constant number of shared read-write registers, when the number of processors $n$ is unknown.
The algorithm has its pseudocode in Figure~\ref{alg:common-bounded-MC}.
To make the exposition of this algorithm more modular, we use two procedures \textsc{Estimate-Size} and \textsc{Extend-Names}.
The pseudocodes of these procedures are given in Figures~\ref{proc:estimate-size} and~\ref{proc:extend-names}, respectively.
The private variables in the pseudocode in Figure~\ref{alg:common-bounded-MC} have the following meaning:  \texttt{size} is an approximation of the number of processors~$n$,  and \texttt{number-of-bins} determines the size of the range of bins we throw conceptual balls into.

The main task of procedure \textsc{Estimate-Size} is to produce an estimate of the number $n$ of processors.
Procedure \textsc{Extend-Names} is iterated multiple times, each iteration is intended to assign names to a group of processors.
This is accomplished by the processors selecting integer values at random, interpreted as throwing balls into bins, and verifying for collisions.
Each selection of a bin is followed by a collision detection.
A ball placement without a detected collision results in a name assigned, otherwise the involved processors try again to throw balls into a range of bins.
The effectiveness of the resulting algorithm hinges of calibrating the number of bins to the expected number of balls to be thrown.

\Paragraph{Balls into bins for the first time.}

The role of procedure \textsc{Estimate-Size}, when called by algorithm \textsc{Common-Bounded-MC}, is to estimate the unknown number of processors~$n$, which is returned as \texttt{size}, to assign a value to variable \texttt{number-of-bins}, and assign values to each private variable \texttt{bin}, which indicates the number of a selected bin in the range $[1,\texttt{number-of-bins}]$.
The procedure tries consecutive values of $k$ as approximations of $\lg n$.
For a given~$k$, an experiment is carried out to throw $n$ balls into $k2^k$ bins.
The execution stops when the number of occupied bins is at most~$2^k$, and then $3\cdot 2^k$ is treated as an approximation of~$n$ and $k2^k$ is the returned number of bins.


\begin{figure}[t]
\rule{\textwidth}{0.75pt}

\F 
\textbf{Procedure} \textsc{Estimate-Size} 

\rule{\textwidth}{0.75pt}
\begin{center}
\begin{minipage}{\pagewidth}
\begin{description}
\item[\rm initialize] $k\gets 2$ 
\hfill /$\ast$ \ initial approximation of $\lg n$ \ $\ast$/
\item[\tt repeat] \ 
\begin{description}
\item[$k$\!]$\gets k+1$ 
\item[\tt bin$_v$]$\gets$ random integer in $[1,k\, 2^k]$
\item[\rm initialize] $\texttt{Nonempty-Bins}\gets 0$
\item[\tt for] $i\gets 1$ \texttt{to}  $k\,2^k$ \texttt{do} 
\begin{description}
\item[\texttt{if}] $\texttt{bin}_v = i$ \texttt{then}
\begin{description}
\item[\tt Nonempty-Bins]$\!\gets \texttt{Nonempty-Bins}+1$ 
\end{description}
\end{description}
\end{description}
\item[\tt until] $\texttt{Nonempty-Bins}\le 2^k$
\item[\tt return] $(3\cdot 2^k,k\,2^k)$
\hfill /$\ast$ \ $3\cdot 2^k$ is \texttt{size},  $k\,2^k$ is \texttt{number-of-bins}\ $\ast$/
\end{description}
\end{minipage}
\FFF

\rule{\textwidth}{0.75pt}

\parbox{\captionwidth}{\caption{\label{proc:estimate-size}
A pseudocode for a processor $v$ of a Common PRAM.
This procedure is invoked by algorithm \textsc{Common-Bounded-MC} in Figure~\ref{alg:common-bounded-MC}.
The variable \texttt{Nonempty-Bins} is shared.
}}
\end{center}
\end{figure}


\begin{lemma}
\label{lem:estimate-size}

For $n\ge 20$ processors, procedure \textsc{Estimate-Size} returns an estimate \texttt{size} of $n$  such that the inequality $\texttt{size}<6n$ holds with certainty and the inequality $n< \texttt{size}$ holds with probability $1-2^{-\Omega(n)}$.
\end{lemma}

\begin{proof}
The procedure returns $3\cdot 2^k$, for some integer $k>0$.
We interpret selecting of values for variable \texttt{bin} in an iteration of the main repeat-loop as throwing $n$ balls into $k2^k$ bins; here $k=j+2$ in the $j$th iteration of this loop, because the smallest value of $k$ is~$3$.
Clearly, $n$ is an upper bound on the number of occupied bins.

If $n$ is a power of~$2$, say $n=2^i$, then the procedure terminates by the time $i=k$, so that $2^k<2^{i+1}=2n$.
Otherwise, the maximum possible $k$ equals $\lceil \lg n \rceil$, because $2^{\lfloor \lg n \rfloor}<n< 2^{\lceil \lg n \rceil}$.
This gives $2^{\lceil \lg n \rceil} = 2^{\lfloor \lg n \rfloor+1}<2n$.
We obtain that the inequality $2^k<2n$ occurs with certainty, and so does $3\cdot 2^k<6n$ as well.

Now we estimate the lower bound on $2^k$.
Consider $k$ such that $2^k\le \frac{n}{3}$.
Then $n$ balls fall into at most $2^k$ bins with probability that is at most
\begin{equation}
\label{eqn:common-bounded-one}
\binom{k2^k}{2^k} \Bigl(\frac{2^k}{k2^k}\Bigr)^{n}
\le
\Bigl(\frac{ek2^k}{2^k}\Bigr)^{2^k}\cdot \frac{1}{k^n}
=
(ek)^{2^k} k^{-n}
=
e^{2^k}k^{2^k-n}
\le
e^{n/3}k^{-2n/3}
\ .
\end{equation}
The right-hand side of~\eqref{eqn:common-bounded-one} is at most $e^{-n/3}$ when the inequality  $k>e$ holds.
The smallest $k$ considered in the pseudocode in Figure~\ref{proc:estimate-size} is $k=3>e$.
The inequality $k>e$  is consistent with $2^k\le \frac{n}{3}$  when $n\ge 20$.
The number of possible values for $k$ is $\cO(\log n)$, so the probability of the procedure returning for $2^k\le \frac{n}{3}$ is $e^{-n/3}\cdot \cO(\log n)= 2^{-\Omega(n)}$.
\end{proof}

Procedure \textsc{Extend-Names}'s behavior can also be interpreted as throwing balls into bins, where a processor $v$'s ball is in a bin $x$ when \texttt{bin}$_v = x$.
The procedure first verifies the suitable range of bins $[1,\texttt{number-of-bins}]$ for collisions.
A verification for collisions takes either just a constant time or $\Theta(\log n)$ time.

A constant verification occurs when there is no ball in the considered bin~$i$, which is verified when  the line ``\texttt{if}  $\texttt{bin}_x=i$ for some processor~$x$'' in the pseudocode in Figure~\ref{proc:extend-names} is to be executed.
Such a verification is performed by using a shared register initialized to~$0$, into which all processors~$v$ with $\texttt{bin}_v=i$ write~$1$, then all the processors read this register, and if the outcome of reading is~$1$ then all write~$0$ again, which indicates that there is at least one ball in the bin, otherwise there is no ball.


\begin{figure}[t]
\rule{\textwidth}{0.75pt}

\F 
\textbf{Procedure} \textsc{Extend-Names}

\rule{\textwidth}{0.75pt}
\begin{center}
\begin{minipage}{\pagewidth}
\begin{description}
\item[\rm initialize] $\texttt{Collision-Detected} \gets \texttt{collision}_v\gets \texttt{false}$ 
\item[\tt for] $i\gets 1$ \texttt{to  \texttt{number-of-bins} do} 
\begin{description}
\item[\tt if]  $\texttt{bin}_x=i$ for some processor $x$ \texttt{then}
\begin{description}
\item[\tt if] $\texttt{bin}_v=i$ \texttt{then}
\begin{description}
\item[\tt for] $j\gets 1$ \texttt{to} $\beta \, \lg \texttt{size}$ \ \texttt{do} 
\begin{description}
\item[\tt if] \textsc{Verify-Collision}   \texttt{then}
\begin{description}
\item[\texttt{Collision-Detected}]\!$\gets \texttt{collision}_v\gets \texttt{true}$
\end{description}
\end{description}
\item[\tt if] $\texttt{not collision}_v$ \texttt{then}
\begin{description}
\item[\tt Last-Name]$\!\gets \texttt{Last-Name}+1$
\item[\tt name$_v$] $ \gets  \texttt{Last-Name}$ 
\item[\tt bin$_v$] $\gets 0$
\end{description}
\end{description}
\end{description}
\end{description}
\item[\tt if]$(\texttt{number-of-bins} > \texttt{size})$  \texttt{then}
\begin{description}
\item[\tt number-of-bins]$\gets \texttt{size}$
\end{description}
\item[\tt if]$\texttt{collision}_v$ \texttt{then}
\begin{description}
\item[\texttt{bin}$_v$]$\!\gets$ random integer in  $[1,\texttt{number-of-bins}]$
\end{description}
\end{description}
\end{minipage}
\FFF

\rule{\textwidth}{0.75pt}

\parbox{\captionwidth}{\caption{\label{proc:extend-names}
A pseudocode for a processor~$v$ of a Common PRAM.
This procedure invokes procedure \textsc{Verify-Collision}, whose pseudocode is in Figure~\ref{proc:verify-collision}, and is itself invoked by algorithm \textsc{Common-Bounded-MC}  in Figure~\ref{alg:common-bounded-MC}.
The  variables \texttt{Last-Name} and \texttt{Collision-Detected} are shared.
The private variable \texttt{name} stores the acquired name.
The constant $\beta>0$ is to be determined in analysis.
}}
\end{center}
\end{figure}

A logarithmic-time verification of collision occurs when there is some ball in the corresponding bin.
This triggers calling procedure \textsc{Verify-Collision} precisely $\beta\lg n$ times; notice that this procedure has the default parameter~$1$, as only one bin is verified at a time.
Ultimately, when a collision is not detected for some processor $v$ whose ball is the bin, then this processor increments \texttt{Last-Name} and assigns its new value as a tentative name. 
Otherwise, when a collision is detected, processor $v$ places its ball in a new bin when the last line in Figure~\ref{proc:extend-names} is executed.

To prepare for the next round of throwing balls, the variable \texttt{number-of-bins} may be reset.
During one iteration of the main repeat-loop of the pseudocode of algorithm \textsc{Common-Bounded-MC} in Figure~\ref{alg:common-bounded-MC}, the number of bins is first set to a value that is $\Theta(n\log n)$ by procedure \textsc{Estimate-Size}.
Immediately after that, it is reset to $\Theta(n)$ by the first call of procedure  \textsc{Extend-Names}, in which the instruction $\texttt{number-of-bins}\gets \texttt{size}$ is performed.
Here, we need to notice that $\texttt{number-of-bins}=\Theta(n\log n)$ and $\texttt{size}=\Theta(n)$, by the pseudocodes in Figures~\ref{proc:estimate-size} and~\ref{alg:common-bounded-MC} and Lemma~\ref{lem:estimate-size}.

\Paragraph{Balls into bins for the second time.}

In the course of analysis of performance of procedure \textsc{Extend-Names}, we consider a balls-into-bins process; we call it simply the \emph{ball process}.
It proceeds through stages so that in a stage we have a number of balls which we throw into a number of bins.
The sets of bins used in different stages are disjoint.
The number of balls and bins used in a stage are as determined in the pseudocode in Figure~\ref{proc:extend-names}, which means that there are $n$ balls and the numbers of bins are as determined by an execution of procedure \textsc{Estimate-Size}, that is, the first stage uses \texttt{number-of-bins} bins and subsequent stages use \texttt{size} bins, as returned by \textsc{Estimate-Size}.

The only difference between the ball process and the actions of procedure \textsc{Extend-Names} is that collisions are detected with certainty in the ball process rather than being tested for.
In particular, the parameter $\beta$ is not involved in the ball process (nor in its name).
The ball process terminates in the first stage in which no multiple bins are produced, so that there are no collisions among the balls.


\begin{lemma}
\label{lem:balls-common-bounded-MC}

The ball process results in all balls ending single in their bins and the number of times a ball is thrown, summed over all the stages, being~$\cO(n)$, both events occurring with probability $1-n^{-\Omega(\log n)}$.
\end{lemma}

\begin{proof}
The argument leverages the property that, in each stage, the number of bins exceeds the number of balls by at least a logarithmic factor.
We will denote the number of bins in a stage by~$m$.
This number will take on two values, first $m=k2^k$ returned as \texttt{number-of-bins}  by procedure \textsc{Estimate-Size} and then $m=3\cdot 2^k$ returned as \texttt{size} by the same procedure \textsc{Estimate-Size}, for $k>3$.
Because $m = k 2^k$ in the first stage, and also $\texttt{size} = 3 \cdot 2^k > n$, by Lemma~\ref{lem:estimate-size}, we obtain that $m > \frac{n}{3} \lg \frac{n}{3}$ in the first stage, and that $m$ is at least $n$ in the following stages, with probability exponentially close to~$1$.

In the first stage, we throw $\ell_1=n$ balls into at least $m=\frac{n}{3} \lg \frac{n}{3}$ bins, with large probability.
Conditional on the event that there are at least these many bins, the probability that a given ball ends the stage single in a bin is
\[
m\cdot \frac{1}{m} \Bigl(1-\frac{1}{m}\Bigr)^{\ell_1-1} 
\ge
1-\frac{\ell_1-1}{m}
\ge
1-\frac{n-1}{\frac{n}{3} \lg \frac{n}{3}}
\ge 1-\frac{4}{\lg n}
\ ,
\]
for sufficiently large $n$, where we used the Bernoulli's inequality.
Let $Y_1$ be the number of singleton bins in the first stage.
The expectancy of $Y_1$ satisfies 
\[
\mE[Y_1] \ge \ell_1 \Bigl(1-\frac{4}{\lg n}\Bigr)
\ .
\]
To estimate the deviation of $Y_1$ from its expected value $\mE[Y_1] $ we use the bounded differences inequality~\cite{McDiarmid89,MitzenmacherUpfal-book05}.
Let $B_j$ be the bin of ball $b_j$, for $1\le j\le \ell_1$.
Then $Y_1$ is of the form $Y_1=h(B_1,\ldots,B_{\ell_1})$, where $h$ satisfies the Lipschitz condition with constant~$2$, because moving one ball to a different bin results in changing the value of $h$ by at most $2$ with respect to the original value.
The bounded-differences inequality specialized to this instance is as follows, for any  $d>0$:
\begin{equation}
\label{eqn-E-Y-1}
\Pr(Y_1 \le \mE [Y_1] -d \sqrt{\ell_1} )
\le
\exp(-d^2/8)
\ .
\end{equation}
We employ $d=\lg n$, which makes the right-hand side of~\eqref{eqn-E-Y-1} asymptotically equal to $n^{-\Omega(\log n)}$. 
The number of balls $\ell_2$ eligible for the second stage can be estimated as follows, this bound holding with probability $1-n^{-\Omega(\log n)}$: 
\begin{equation}
\label{eqn:ell-2}
\ell_2
\le
\frac{4\ell_1}{\lg n} + \lg n \sqrt{\ell_1}
=
\frac{4\ell_1}{\lg n} \Bigl(1+\frac{\lg^2 n}{4\sqrt{\ell_1}}\Bigr) 
\le
\frac{5n}{\lg n}  
\ ,
\end{equation}
for sufficiently large $n$.

In the second stage, we throw $\ell_2$ balls into $m\ge n$ bins, with large probability.
Conditional on the bound~\eqref{eqn:ell-2} holding, the probability that a given ball ends up  single  in a bin is
\[
m\cdot \frac{1}{m} \Bigl(1-\frac{1}{m}\Bigr)^{\ell_2-1} 
\ge
1-\frac{\ell_2-1}{m}
\ge 1-\frac{5}{\lg n}
\ ,
\]
where we used the Bernoulli's inequality.
Let $Y_2$ be the number of singleton bins in the second stage.
The expectancy of $Y_2$ satisfies 
\[
\mE[Y_2] \ge \ell_2 \Bigl(1-\frac{5}{\lg n}\Bigr)
\ .
\]
To estimate the deviation of $Y_2$ from its expected value $\mE[Y_2]$, we again use the bounded differences inequality, which specialized to this instance is as follows, for any  $d>0$:
\begin{equation}
\label{eqn-E-Y-2}
\Pr(Y_2 \le \mE [Y_2] -d \sqrt{\ell_2} )
\le
\exp(-d^2/8)
\ .
\end{equation}
We again employ $d=\lg n$, which makes the right-hand side of~\eqref{eqn-E-Y-2} asymptotically equal to $n^{-\Omega(\log n)}$.
The number of balls $\ell_3$ eligible for the third stage can be bounded from above as follows, which holds with probability $1-n^{-\Omega(\log n)}$, : 
\begin{equation}
\label{eqn:ell-3}
\ell_3
\le
\frac{5\ell_2}{\lg n} + \lg n \sqrt{\ell_2}
=
\frac{5\ell_2}{\lg n} \Bigl(1+\frac{\lg^2 n}{5\sqrt{\ell_2}}\Bigr) 
\le
\frac{6n}{\lg^2 n}  
\ ,
\end{equation}
for sufficiently large $n$.

Next, we generalize these estimates.
In stages $i$, for $i\ge 2$, among the first $\cO(\log n)$ ones, we throw balls into $m\ge n$ bins with large probability.
Let $\ell_i$ be the number of balls eligible for such a stage~$i$. 
We show by induction that $\ell_i$, for $i\ge 3$,  can be estimated as follows:
\begin{equation}
\label{eqn:ell-i}
\ell_{i} \le \frac{6n}{\lg^{2} n} \cdot 4^{3-i}
\end{equation}
with probability $1-n^{-\Omega(\log n)}$.
The estimate~\eqref{eqn:ell-3} provides the base of induction for $i=3$.
In the inductive step, we assume~\eqref{eqn:ell-i}, and consider what happens during stage~$i>3$ in order to estimate the number of balls eligible for the next stage $i+1$.

In stage~$i$, we throw $\ell_i$ balls into $m\ge n$ bins, with large probability.
Conditional on the bound~\eqref{eqn:ell-i}, the probability that a given ball ends up single in a bin is
\[
m\cdot \frac{1}{m} \Bigl(1-\frac{1}{m}\Bigr)^{\ell_i-1} 
\ge
1-\frac{\ell_i-1}{m}
\ge 
1-\frac{6\cdot 4^{3-i}}{\lg^{2} n}
\ ,
\]
by the inductive assumption,  where we also used the Bernoulli's inequality.
If $Y_i$ is the number of singleton bins in stage~$i$, then its expectation~$\mE[Y_i]$ satisfies 
\begin{equation}
\label{eqn:one-minus-o-of-one}
\mE[Y_i] \ge \ell_i \Bigl(1-\frac{6\cdot 4^{3-i}}{\lg^{2} n}\Bigr)
\ .
\end{equation}
To estimate the deviation of $Y_i$ from its expected value $\mE[Y_i] $, we again use the bounded differences inequality, which specialized to this instance is as follows, for any  $d>0$:
\begin{equation}
\label{eqn-E-Y-inductive-step}
\Pr(Y_i \le \mE [Y_i] -d \sqrt{\ell_i} )
\le
\exp(-d^2/8)
\ .
\end{equation}
We employ $d=\lg n$, which makes the right-hand side of~\eqref{eqn-E-Y-inductive-step} asymptotically equal to $n^{-\Omega(\log n)}$.
The number of balls $\ell_{i+1}$ eligible for the next stage $i+1$ can be estimated from above in the following way, the estimate holding with probability $1-n^{-\Omega(\log n)}$ : 
\begin{eqnarray*}
\label{eqn:ell-i-plus-1}
\ell_{i+1}
&\le&
\frac{6\cdot 4^{3-i}\cdot \ell_i }{\lg^{2} n} + \lg n \sqrt{\ell_i}\\
&=&
\frac{6\cdot 4^{3-i}\cdot \ell_i }{\lg^{2} n} \Bigl(1+\frac{1}{6} 4^{i-3} \lg^3 n \cdot \ell_i^{-1/2}\Bigr) \\
&\le&
\frac{6\cdot 4^{3-i}}{\lg^{2} n} \cdot \frac{6n}{\lg^{2} n} \cdot 4^{3-i} \cdot  \Bigl(1+\frac{4^{(i-3)/2} \lg^4 n }{6 \sqrt{6n}}\Bigr) \\
&\le&
\frac{6n}{\lg^{2} n} \cdot 4^{3-i} \cdot  \Bigl(\frac{6 \cdot 4^{3-i}}{\lg^{2} n}+\frac{4^{(3-i)/2} \lg^2 n }{ \sqrt{6n}}\Bigr) \\
&\le&
\frac{6n}{\lg^{2} n} \cdot 4^{3-i} \cdot  \Bigl(\frac{6 }{\lg^{2} n}+\frac{ \lg^2 n }{ \sqrt{6n}}\Bigr) \\
&\le&
\frac{6n }{\lg^{2} n}  \cdot 4^{3-i-1}
\ ,
\end{eqnarray*}
for sufficiently large $n$ that does not depend on~$i$.
For the event $Y_i \le \mE [Y_i] -d \sqrt{\ell_i}$ in the estimate~\eqref{eqn-E-Y-inductive-step} to be meaningful, it is sufficient if the following estimate holds:
\[
\lg n \cdot \sqrt{\ell_i}=o(\mE [Y_i])
\ .
\]
This is the case as long as $\ell_i>\lg^3 n$, because $\mE [Y_i]= \ell_i(1+o(1))$ by~\eqref{eqn:one-minus-o-of-one}.

To summarize at this point, as long as $\ell_i$ is sufficiently large, that is, $\ell_i>\lg^3 n$, the number of eligible balls decreases by at least a factor of~$4$ with probability that is at least  $1-n^{-\Omega(\log n)}$.
It follows that the total number of eligible balls, summed over these stages, is $\cO(n)$ with this probability.

After at most $\lg_4 n=\frac{1}{2}\lg n$ such stages, the number of balls becomes at most $\lg^3 n$ with probability $1-n^{-\Omega(\log n)}$.
This number of stages is half of the number of times the for-loop is iterated in the pseudocode in Figure~\ref{alg:common-bounded-MC}.

It remains to consider the stages when $\ell_i\le \lg^3 n$, so that we throw at most $\lg^3 n$ balls into at least $n$ bins.
They all end up in singleton bins with a probability that is at least
\[
\Bigl(\frac{n-\lg^3 n}{n}\Bigr)^{\lg^3 n}
\ge 
\Bigl(1-\frac{\lg^3 n}{n}\Bigr)^{\lg^3 n}
\ge 
1-\frac{\lg^{6} n}{n}
\ ,
\]
by the Bernoulli's inequality.
So the probability of a collision is at most $\frac{\lg^{6} n}{n}$.
One stage without any collision terminates the process.
If we repeat such stages $\frac{1}{2}\lg n$ times, without even removing single balls, then the probability of collisions occurring in all these stages is at most 
\[
\Bigl(\frac{\lg^{6} n}{n}\Bigr)^{\frac{1}{2}\lg n}=n^{-\Omega(\log n)}
\ .
\]
This number of stages is half of the number of times the for-loop  is iterated in the pseudocode in Figure~\ref{alg:common-bounded-MC}.
The number of eligible balls summed over these final stages is at most $\lg^{7} n=o(n)$.
\end{proof}


\begin{figure}[t]
\rule{\textwidth}{0.75pt}

\F 
\textbf{Algorithm} \textsc{Common-Bounded-MC}

\rule{\textwidth}{0.75pt}
\begin{center}
\begin{minipage}{\pagewidth}
\begin{description}
\item[\tt repeat] \ 
\begin{description}
\item[\rm initialize] $\texttt{Last-Name} \gets 0$
\item[\texttt{(size, number-of-bins)}] $\gets$ \textsc{Estimate-Size}
\item[for] $\ell \gets 1$ \texttt{to} $\lg \texttt{size}$ \texttt{do}
\begin{description}
\item[\sc Extend-Names] 
\item[\tt if] \texttt{not Collision-Detected then return}
\end{description}
\end{description}
\end{description}
\end{minipage}
\FFF

\rule{\textwidth}{0.75pt}

\parbox{\captionwidth}{\caption{\label{alg:common-bounded-MC}
A pseudocode for a processor~$v$ of a Common PRAM, where  there is a constant  number of shared memory cells.
Procedures \textsc{Estimate-Size} and \textsc{Extend-Names} have their pseudocodes in Figures~\ref{proc:estimate-size} and~\ref{proc:extend-names}, respectively.
The variables \texttt{Last-Name} and \texttt{Collision-Detected} are shared.
}}
\end{center}
\end{figure}

The following Theorem~\ref{thm:common-bounded-MC} summarizes the performance of algorithm \textsc{Common-Bounded-MC} (see the pseudocode in Figure~\ref{alg:common-bounded-MC}) as a Monte Carlo one.


\begin{theorem}
\label{thm:common-bounded-MC}

Algorithm \textsc{Common-Bounded-MC}  terminates almost surely.
For each $a>0$, there exists $\beta>0$ and $c>0$ such that  the algorithm assigns unique names, works in time at most~$cn\ln n$, and  uses at most $cn\ln n$ random bits, each among these properties holding with probability at least $1-n^{-a}$.
\end{theorem}

\begin{proof}
One iteration of the main repeat-loop suffices to assign names with probability $1-n^{-\Omega(\log n)}$, by Lemma~\ref{lem:balls-common-bounded-MC}.
This means that the probability of not terminating by the $i$th iteration is at most $(n^{- \Omega(\log n)})^i$, which converges to~$0$ with $i$ growing to infinity.

The algorithm returns duplicate names only when a collision occurs that is not detected by procedure \textsc{Verify-Collision}.
For a given multiple bin, one iteration of this procedure does not detect collision with probability at most $1/2$, by Lemma~\ref{lem:verify-collision}.
Therefore $\beta \lg \texttt{size}$ iterations do not detect collision with probability $\cO(n^{-\beta/2})$, by Lemma~\ref{lem:estimate-size}.
The number of nonempty bins ever tested is at most $d n$, for some constant $d>0$, by Lemma~\ref{lem:balls-common-bounded-MC}, with the suitably large probability.
Applying the union bound results in the estimate $n^{-a}$ on the probability of error for sufficiently large $\beta$.

The duration of an iteration of the inner for-loop is either constant, then we call is \emph{short}, or it takes time $\cO(\log \texttt{size})$, then we call it \emph{long}.
First, we estimate the total time spent on short iterations.
This time in the first iteration of the inner for-loop is proportional to \texttt{number-of-bins} returned by procedure \textsc{Estimate-Size}, which is at most $6n\cdot \lg(6n)$, by Lemma~\ref{lem:estimate-size}.
Each of the subsequent iterations takes time proportional to \texttt{size}, which is at most $6n$, again by Lemma~\ref{lem:estimate-size}.
We obtain that the total number of short iterations is $\cO(n\log n)$ in the worst case.
Next, we estimate the total time spent on long iterations.
One such an iteration has time proportional to $\lg \texttt{size}$, which is at most $\lg 6n$ with certainty.
The number of such iterations is at most $d n$ with probability $1-n^{-\Omega(\log n)}$, for some constant $d>0$, by Lemma~\ref{lem:balls-common-bounded-MC}.
We obtain that the total number of long iterations is $\cO(n\log n)$, with the correspondingly large probability.
Combining the estimates for short and long iterations, we obtain $\cO(n\log n)$ as a bound on time of one iteration of the main repeat-loop.
One such an iteration suffices with probability $1-n^{-\Omega(\log n)}$, by Lemma~\ref{lem:balls-common-bounded-MC}.

Throwing one ball uses $\cO(\log n)$ random bits, by Lemma~\ref{lem:estimate-size}.
The number of throws is $\cO(n)$ with the suitably large probability, by Lemma~\ref{lem:balls-common-bounded-MC}.
\end{proof}

Algorithm \textsc{Common-Bounded-MC} is optimal with respect to the following performance metrics: the expected time $\cO(n\log n)$, by Theorem~\ref{thm:lower-bound-memory-common-pram}, the number of random bits $\cO(n\log n)$, by Proposition~\ref{pro:lower-bound-on-random-bits}, and the probability of error $n^{-\cO(1)}$, by Proposition~\ref{pro:probability-of-error}.

\section{Monte Carlo for Common with Unbounded Memory}

\label{sec:common-unbounded-MC}

We consider naming on a Common PRAM in the case when the amount of shared memory is unbounded and the number of processors~$n$ is unknown.
The Monte Carlo algorithm we propose, called \textsc{Common-Unbounded-MC}, is similar to algorithm \textsc{Common-Bounded-MC} in Section~\ref{sec:common-bounded-MC}, in that it involves a randomized experiment to  estimate the number of processors of the PRAM.
Such an experiment is then followed by repeatedly throwing balls into bins, testing for collisions, and throwing again if a collision is detected, until eventually no collisions are detected.

Algorithm \textsc{Common-Unbounded-MC} has its pseudocode given in Figure~\ref{alg:common-unbounded-MC}.
The algorithm is structured as a repeat loop.
An iteration starts by invoking procedure \textsc{Gauge-Size}, whose pseudocode is in Figure~\ref{proc:gauge-size-mc}.
This procedure returns  \texttt{size} as an estimate  of the number of processors~$n$.
Next, a processor chooses randomly a bin in the range $[1,3\texttt{size}]$.
Then it keeps verifying for collisions $\beta\lg \texttt{size}$, in such a manner that when a collision is detected then a new bin is selected form the same range.
After such $\beta\lg \texttt{size}$ verifications and possible new selections of bins, another $\beta\lg \texttt{size}$ verifications follow, but without changing the selected bins.
When no collision is detected in the second segment of $\beta\lg \texttt{size}$ verifications, then this terminates the repeat-loop, which triggers assigning each station the rank of the selected bin, by a prefix-like computation.
If a collision is detected in the second segment of $\beta\lg \texttt{size}$ verifications, then this starts another iteration of the main repeat-loop.

Procedure \textsc{Gauge-Size-MC} returns an estimate of the number~$n$ of processors in the form~$2^k$, for some positive integer~$k$.
It operates by trying various values of~$k$, and, for a considered $k$, by throwing $n$ balls into $2^k$ bins and next counting how many bins contain balls.
Such counting is performed by a prefix-like computation, whose pseudocode is omitted in Figure~\ref{proc:gauge-size-mc}.
The additional parameter $\beta>0$ is a number that affects the probability of underestimating~$n$.

The way in which selections of numbers~$k$ is performed is controlled by function~$r(k)$, which is a parameter.
We will consider two instantiations of this function: one is function~$r(k)=k+1$ and the other is function~$r(k)=2k$.


\begin{figure}[t]
\rule{\textwidth}{0.75pt}

\F 
\textbf{Procedure} \textsc{Gauge-Size-MC}

\rule{\textwidth}{0.75pt}
\begin{center}
\begin{minipage}{\pagewidth}
\begin{description}
\item[$k\gets 1$]\ 
\item[\tt repeat] \ 
\begin{description}
\item[$k\gets r(k)$] 
\item[$\texttt{bin}_v \gets$] random integer in $[1,2^k]$ 
\end{description}
\item[\tt until] the number of selected values of variable \texttt{bin} is $\le 2^k/\beta$ 
\item[\tt return]$(\, \lceil 2^{k+1}/\beta\rceil \,)$
\end{description}
\end{minipage}
\FFF

\rule{\textwidth}{0.75pt}

\parbox{\captionwidth}{\caption{\label{proc:gauge-size-mc}
A pseudocode for a processor~$v$ of a Common PRAM, where the number of shared memory cells is unbounded.
The constant~$\beta>0$ is the same parameter as in Figure~\ref{alg:common-unbounded-MC}, and an increasing function~$r(k)$ is also a parameter.}}
\end{center}
\end{figure}

\begin{lemma}
\label{lem:gauge-size-mc}

If $r(k)=k+1$ then the value of \texttt{size} as returned by \textsc{Gauge-Size-MC} satisfies $\texttt{size}\le  2 n$ with certainty and the inequality $\texttt{size}\ge n$ holds with probability  $1-\beta^{-n/3}$.

If $r(k)=2k$ then the value of \texttt{size} as returned by \textsc{Gauge-Size-MC} satisfies $\texttt{size}\le 2\beta n^2$ with certainty and $\texttt{size}\ge \beta n^2/2$ with probability $1-\beta^{-n/3}$.
\end{lemma}

\begin{proof}
We model procedure's execution by an experiment of throwing $n$ balls into $2^k$ bins.
If the parameter function $r(k)$ is $r(k)=k+1$ then this results in trying all possible consecutive values of~$k$, starting from $k=2$, so that $k=i+1$ in the $i$th iteration of the repeat-loop.
If the parameter function $r(k)$ is $r(k)=2k$ then $k$ takes on only the powers of~$2$.

There are at most $n$ bins occupied in any such an experiment.
Therefore, the procedure returns by the time the inequality $2^k/\beta \ge n$ holds, where $k$ determines the range of bins.
It follows that if $r(k)=k+1$ then the returned value $\lceil 2^{k+1}/\beta\rceil$ is at most $2n$.
If $r(k)=2k$ then the worst error in estimating occurs when $2^i/\beta = n-1$ for some~$i$ that is a power of~$2$.
Then the returned value is $2^{2i}/\beta=(\beta (n-1))^2/\beta$, which is at most $2\beta n^2$, this occurring with probability  $1-\beta^{-n/3}$.

Given $2^k$ bins, we estimate the probability that the number of occupied bins is at most $2^k/\beta$.
It is
\[
\binom{2^k}{2^k/\beta}\Bigl(\frac{2^k/\beta}{2^k}\Bigr)^n
\le
\Bigl(\frac{2^k e}{2^k/\beta}\Bigr)^{2^k/\beta}\cdot \frac{1}{\beta^n}
=
(e\beta)^{2^k/\beta}\cdot \beta^{-n}
\ .
\]
Next, we identify a range of values of $k$ for which this probability is exponentially close to~$0$ with respect to~$n$.

To this end, let $0<\rho<1$ and let us consider the inequality 
\begin{equation}
\label{eqn:beta-rho}
(e\beta)^{2^k/\beta}\cdot \beta^{-n}<\rho^n
\ .
\end{equation} 
It is equivalent to the following one
\[
\frac{2^k}{\beta}(1+\ln \beta)-n\ln \beta < n\ln \rho
\ ,
\]
by taking logarithms of both sides.
This in turn is equivalent to
\begin{equation}
\label{eqn:from-rho-to-beta}
\frac{2^k}{\beta}(1+\ln \beta)< n\Bigl(\ln \beta -\ln \frac{1}{\rho}\Bigr)
\ .
\end{equation}
Let us choose $\rho=\beta^{-1/2}$ in~\eqref{eqn:from-rho-to-beta}.
Then~\eqref{eqn:beta-rho} specialized to this particular $\rho$ is equivalent to the following inequality
\[
\frac{2^k}{\beta}(1+\ln \beta) < n  \frac{\ln \beta}{2}
\ .
\]
This in turn leads to the estimate
\[
2^k 
< 
n\cdot \frac{\ln \beta}{2} \cdot \frac{\beta}{1+\ln \beta} 
< 
\frac{\beta}{2}\cdot n
\ ,
\]
which means $2^{k+1}/\beta< n$.
When $k$ satisfies this inequality then the probability of returning is at most $\beta^{-n/2}$.
There are $\cO(\log n)$ such values of $k$ considered by the procedure, so it returns for one of them with probability at most 
\[
\cO(\log n)\cdot \beta^{-n/2} <\beta^{-n/3}
\ ,
\] 
for sufficiently large~$n$.
Therefore, with probability at least $1-\beta^{-n/3}$, the returned value $\lceil 2^{k+1}/\beta\rceil$ is at least as large as determined the first considered $k$ that satisfies $2^{k+1}/\beta\ge n$.

If $r(k)=k+1$ then all the possible exponents $k$ are considered, so the returned value $\lceil 2^{k+1}/\beta\rceil$ is at least $n$ with probability $1-\beta^{-n/3}$.
If $r(k)=2k$ then the worst error of estimating $n$ occurs when $2^{i+1}/\beta= n-1$ for some $i$ that is a power of~$2$.
Then the returned value is 
\[
2^{2i+1}/\beta= 2\cdot (\beta(n-1)/2)^2/\beta
\ , 
\]
which is is at least $\beta n^2/2$, this occurring  with probability $1-\beta^{-n/3}$.
\end{proof}

We discuss performance of algorithm \textsc{Common-Unbounded-MC} (see the pseudocode in Figure~\ref{alg:common-unbounded-MC}) by referring to analysis of a related algorithm \textsc{Common-Unbounded-LV} given in~Section~\ref{sec:common-unbounded-LV}.
We consider a \emph{$\beta$-process with verifications}, which is  defined as follows.
The process proceeds through stages.
The first stage starts with placing $n$ balls into $3\,\texttt{size}$ bins.
For each of the subsequent stages, for all multiple bins and for each ball in such a bin, we perform a Bernoulli trial with  the probability~$\frac{1}{2}$ of success, which represents the outcome of procedure \textsc{Verify-Collision}.
A success in a trial is referred to as a \emph{positive verification} otherwise it is a \emph{negative} one.
If at least one positive verification occurs for a ball in a multiple bin then all the balls in this bin are relocated in this stage to bins selected uniformly at random and independently for each such a ball, otherwise the balls stay put in this bin until the next stage.
The process terminates when all balls are single in their bins.


\begin{figure}[t]
\rule{\textwidth}{0.75pt}

\F 
\textbf{Algorithm} \textsc{Common-Unbounded-MC}

\rule{\textwidth}{0.75pt}
\begin{center}
\begin{minipage}{\pagewidth}
\begin{description}
\item[\tt repeat] \ 
\begin{description}
\item[$\texttt{size}\gets$] \textsc{Gauge-Size}
\item[$\texttt{bin}_v \gets$] random integer in $[1,3 \, \texttt{size}]$ 
\item[\tt for] $i\gets 1$ \texttt{to} $\beta \lg \texttt{size}$ \texttt{do} 
\begin{description}
\item[\tt if] \textsc{Verify-Collision}\,$(\texttt{bin}_v)$ \texttt{then} 
\begin{description}
\item[$\texttt{bin}_v\gets$] random number in $[1,3 \, \texttt{size}]$
\end{description}
\end{description}
\item[\tt Collision-Detected]$\gets$ \texttt{false}
\item[\tt for] $i\gets 1$ \texttt{to} $\beta \lg \texttt{size}$ \texttt{do} 
\begin{description}
\item[\tt if] \textsc{Verify-Collision}\,$(\texttt{bin}_v)$ \texttt{then} 
\begin{description}
\item[\tt Collision-Detected]$\gets$ \texttt{true} 
\end{description}
\end{description}
\end{description}
\item[\tt until] \texttt{not Collision-Detected} 
\item[\tt name$_v$]\!$\gets$ the rank of $\texttt{bin}_v$ among selected bins
\end{description}
\end{minipage}
\FFF

\rule{\textwidth}{0.75pt}

\parbox{\captionwidth}{\caption{\label{alg:common-unbounded-MC}
A pseudocode for a processor~$v$ of a Common PRAM, where the number of shared memory cells is unbounded.
The constant~$\beta>0$ is a parameter impacting the probability of error.
The private variable \texttt{name} stores the acquired name.
}}
\end{center}
\end{figure}

\begin{lemma}
\label{lem:beta-verifications}

For any number $a>0$, there exists $\beta> 0$ such that the $\beta$-process with verifications  terminates within $\beta \lg n$ stages with all of them comprising the total of $\cO(n)$ ball throws with probability at least~$1-n^{-a}$.
\end{lemma}

\begin{proof}
We use the analysis of a ball process relevant to Common PRAM with unbounded memory given in~Section~\ref{sec:common-unbounded-LV}.
The constant~$3$ determining our $\beta$-process with verifications  corresponds to $1+\beta$ in~Section~\ref{sec:common-unbounded-LV}.
The corresponding $\beta$-process in verifications considered in~Section~\ref{sec:common-unbounded-LV} is defined by referring to known~$n$.
We use the approximation \texttt{size} instead, which is at least as large as $n$ with probability  $1-\beta^{-n/3}$, by Lemma~\ref{lem:gauge-size-mc} just proved.
By Lemma~\ref{lem:ball-process-with-verifications-terminates}, our $\beta$-process with verifications does not terminate within $\beta\lg n$ stages when $\texttt{size}\ge n$ with probability at most $n^{-2a}$ and the inequality $\texttt{size}\ge n$ does not hold with probability at most $\beta^{-n/3}$.
Therefore the conclusion we want to prove does not hold with probability at most $n^{-2a}+\beta^{-n/3}$, which is at most $n^{-2a}$ for sufficiently large~$n$.
\end{proof}

The following Theorem summarizes the performance of algorithm \textsc{Common-Unbounded-MC} (see the pseudocode in Figure~\ref{alg:common-unbounded-MC}) as a Monte Carlo one.
Its proof relies on mapping an execution of the $\beta$-process with verifications on executions of algorithm \textsc{Common-Unbounded-MC} in a natural manner.


\begin{theorem}
\label{thm:common-unbounded-MC}

Algorithm \textsc{Common-Unbounded-MC} terminates almost surely, for a sufficiently large~$\beta$.
For each $a>0$, there exists $\beta>0$ and $c>0$ such that  the algorithm assigns unique names and has the following additional properties with probability $1-n^{-a}$.
If $r(k)=k+1$ then at most $cn$ memory cells are ever needed, $c n\ln^2 n$ random bits are ever generated, and the algorithm terminates in time $\cO(\log^2 n)$.
If $r(k)=2k$ then at most $cn^2$ memory cells are ever needed, $c n\ln n$ random bits are ever generated, and the algorithm terminates in time $\cO(\log n)$.
\end{theorem}
  
\begin{proof}
For a given $a>0$, let us take $\beta$ that exists by Lemma~\ref{lem:beta-verifications}.
When the $\beta$-process with verifications terminates then this models assigning unique names by the algorithm.
It follows that one iteration of the repeat-loop results in the algorithm terminating with proper names assigned with probability $1-n^{-a}$.
One iteration of the main repeat-loop does not result in termination with probability at most~$n^{-a}$, so $i$ iterations are not sufficient to terminate with probability at most ~$n^{-ia}$.
This converges to~$0$ with increasing~$i$ so the algorithm terminates almost surely.

The performance metrics rely mostly on Lemma~\ref{lem:gauge-size-mc}.
We consider two cases, depending on which function $r(k)$ is used.

If $r(k)=k+1$ then procedure \textsc{Gauge-Size-MC} considers all the consecutive values of~$k$ up to~$\lg n$, and for each such~$k$, throwing a ball requires $k$ random bits.
We obtain that procedure \textsc{Gauge-Size-MC} uses $\cO(n\log^2 n)$ random bits.
Similarly, to compute the number of selected values in an iteration of the main repeat-loop of this procedure takes time $\cO(k)$, for the corresponding~$k$, so this procedure takes $\cO(\log^2 n)$ time.
The value of \texttt{size}  satisfies $\texttt{size}\le  2 n$ with certainty.
Therefore, $\cO(n)$ memory registers are ever needed, while one throw of a ball uses $\cO(\log n)$ random bits, after \texttt{size} has been computed.
It follows that one iteration of the main repeat-loop of the algorithm, after procedure \textsc{Gauge-Size-MC} has been completed, uses $\cO(n\log n)$ random bits, by Lemmas~\ref{lem:gauge-size-mc} and~\ref{lem:beta-verifications}, and takes $\cO(\log n)$ time.
Since one iteration of the main repeat-loop suffices with probability $1-n^{-a}$, the overall time is dominated by the time performance of procedure \textsc{Gauge-Size-MC}.

If $r(k)=2k$ then procedure \textsc{Gauge-Size-MC} considers all the consecutive powers of~$2$ as values of $k$ up to $\lg n$, and for each such~$k$, throwing a ball requires $k$ random bits.
Since the values $k$ form a geometric progression, procedure \textsc{Gauge-Size-MC} uses $\cO(\log n)$ random bits per processor.
Similarly, to compute the number of selected values in an iteration of the main repeat-loop of this procedure takes time $\cO(k)$, for the corresponding~$k$ that increase geometrically, so this procedure takes $\cO(\log n)$ time.
The value of \texttt{size}  satisfies $\texttt{size}\le  2 n$ with certainty.
By Lemma~\ref{lem:gauge-size-mc}, $\cO(n^2)$ memory registers are ever needed, so one throw of a ball uses $\cO(\log n)$ random bits.
One iteration of the main repeat-loop, after procedure \textsc{Gauge-Size-MC} has been completed, uses $\cO(n\log n)$ random bits, by Lemmas~\ref{lem:gauge-size-mc} and~\ref{lem:beta-verifications}, and takes $\cO(\log n)$ time.
\end{proof}

The instantiations of algorithm \textsc{Common-Unbounded-MC} are close to optimality with respect to some of the performance metrics we consider, depending on whether $r(k)=k+1$ or $r(k)=2k$.
If $r(k)=k+1$ then the algorithm's use of shared memory would be optimal if its time were $\cO(\log n)$, by Theorem~\ref{thm:lower-bound-memory-arbitrary-pram}, but it misses space optimality by at most a logarithmic factor, since the algorithm's time is $\cO(\log^2 n)$.
Similarly, for this case of $r(k)=k+1$, the number of random bits ever generated $\cO(n\log^2 n)$ misses optimality by at most a logarithmic factor, by Proposition~\ref{pro:lower-bound-on-random-bits}.
In the other case of $r(k)=2k$, the expected time $\cO(\log n)$ is optimal, by Theorem~\ref{thm:log-n-lower-bound}, the expected number of random bits  $\cO(n\log n)$ is optimal, by Proposition~\ref{pro:lower-bound-on-random-bits},  and the probability of error~$n^{-\cO(1)}$ is optimal, by Proposition~\ref{pro:probability-of-error}, but the amount of used shared memory misses optimality by at most a polynomial factor, by Theorem~\ref{thm:log-n-lower-bound}.

\section{Conclusion}

\label{sec:conclusion}

We considered the naming problem for the anonymous synchronous PRAM when the number of processors $n$ is known.
We gave Las Vegas algorithms for  four variants of the problem, which are determined by the suitable restrictions on concurrent writing and the amount of shared memory.  
Each of these algorithms is provably optimal for its case with respect to the natural performance metrics such as expected time (as determined by the amount of shared memory) and expected number of used random bits.

We also considered four variants of the naming problem for an anonymous PRAM, when the number of processors $n$ is unknown, and developed Monte Carlo naming algorithms for each of them.
The two algorithms for a bounded number of shared registers are provably optimal with respect to  the following three performance metrics: expected time, expected number of generated random bits and probability of error.
It is an open problem to develop Monte Carlo algorithms for Arbitrary and Common PRAMs for the case when the amount of shared memory is unbounded, such that they are simultaneously asymptotically optimal with respect to these same three performance metrics: the expected time, the expected number of generated random bits and the probability of error.

The algorithms we gave cover the ``boundary'' cases of the model.
One case is about a minimum amount of shared memory, that is, when  only a constant number of  shared memory cells are available.
The other case is about a minimum expected running time, that is, when the expected running time  is $\cO(\log n)$; such performance requires a number of shared registers that grows unbounded  with~$n$.
It would be interesting to have the results of this paper generalized by investigating naming on a PRAM when  the number of processors and the number of shared registers are independent parameters of the model.


\bibliographystyle{abbrv}

\bibliography{anonym-pram}

\end{document}